\documentclass[english]{extarticle}
\usepackage[T1]{fontenc}
\usepackage[latin9]{inputenc}
\usepackage{geometry}
\usepackage{float}
\usepackage{amsmath}
\usepackage{amsthm}
\usepackage{amssymb}
\usepackage{thm-restate}

\makeatletter

\floatstyle{ruled}
\newfloat{algorithm}{tbp}{loa}
\providecommand{\algorithmname}{Algorithm}
\floatname{algorithm}{\protect\algorithmname}

\theoremstyle{plain}
\newtheorem{thm}{\protect\theoremname}[section]
  \theoremstyle{definition}
  \newtheorem{defn}[thm]{\protect\definitionname}
  \theoremstyle{plain}
  \newtheorem{fact}[thm]{\protect\factname}
  \theoremstyle{plain}
  \newtheorem{lem}[thm]{\protect\lemmaname}
  \theoremstyle{plain}
  
    \theoremstyle{plain}
  \newtheorem{cor}[thm]{\protect\corollaryname}

\usepackage{algorithm}\usepackage{algorithmic}
\usepackage{color}
\usepackage{enumitem}

\definecolor{marekgreen}{RGB}{0,185,0}

\ifdefined\DEBUG
\definecolor{marekgreen}{RGB}{0,185,0}
\definecolor{orange}{RGB}{255,128,0}

  \def\rem#1{{\marginpar{\raggedright\scriptsize #1}}}
  \newcommand{\micr}[1]{\rem{\textcolor{blue}{$\bullet$ #1}}}
  \newcommand{\marr}[1]{\rem{\textcolor{marekgreen}{$\bullet$ #1}}}

\else

  \newcommand{\marr}[1]{}
  \newcommand{\micr}[1]{}
  
\fi

\makeatother

\usepackage{babel}
  \providecommand{\definitionname}{Definition}
  \providecommand{\factname}{Fact}
  \providecommand{\lemmaname}{Lemma}
  \providecommand{\propositionname}{Proposition}
\providecommand{\theoremname}{Theorem}
  \providecommand{\corollaryname}{Corollary}

\begin{document}
\global\long\def\adj{\mbox{\footnotesize Adj}}

\global\long\def\chr#1{\mathbf{1}_{#1}}

\global\long\def\br#1{\left( #1 \right)}

\global\long\def\brq#1{\left[ #1 \right]}

\global\long\def\brw#1{\left\{  #1\right\}  }

\global\long\def\cut#1{\partial#1 }

\global\long\def\ex#1{\mathbb{E}\left[#1\right]}

\global\long\def\E{\mathbb{E}}

\global\long\def\exls#1#2{\mathbb{E}_{#1}\left[#2\right]}

\global\long\def\set#1{\left\{  #1\right\}  }

\global\long\def\ind#1{\mathbf{1}\left[ #1 \right]}

\global\long\def\st#1{[#1] }

\global\long\def\opstyle#1{\mathbb{#1}}

\global\long\def\ex#1{\mathbb{E}\left[#1\right]}

\global\long\def\xp#1{\mathbb{E}\left[#1\right]}

\global\long\def\exls#1#2{\opstyle{\opstyle E}_{#1}\left[ #2 \right]}

\global\long\def\size#1{\left|#1\right|}

\global\long\def\set#1{\left\{  #1\right\}  }

\global\long\def\adj{\mbox{\footnotesize Adj}}

\global\long\def\st#1{[#1] }

\global\long\def\indi#1{\chi\brq{#1}}

\global\long\def\evalat#1#2{ #1 \Big|_{#2}}

\global\long\def\Df#1#2{\frac{\partial#1}{\partial#2}}

\global\long\def\ex#1{\mathbb{E}\left[#1\right]}

\global\long\def\prls#1#2{\opstyle P_{#1}\left[ #2 \right]}
 \global\long\def\pr#1{\opstyle P \left[ #1 \right]}

\global\long\def\exls#1#2{\opstyle{\opstyle E}_{#1}\left[ #2 \right]}

\global\long\def\excondls#1#2#3{\mathbb{E}_{#1}\brq{\left.#2\right|#3}}

\global\long\def\prcondls#1#2#3{\mathbb{P}_{#1}\brq{\left.#2\right|#3}}

\global\long\def\size#1{\left|#1\right|}

\global\long\def\setst#1#2{\left\{  #1\ \left|\ #2\right.\right\}  }

\global\long\def\setstX#1#2{\left\{  \left.#1\right|#2\right\}  }

\global\long\def\setstcol#1#2{\left\{  #1:#2\right\}  }

\global\long\def\set#1{\left\{  #1\right\}  }

\global\long\def\adj#1{\delta\br{#1}}

\global\long\def\st#1{[#1] }

\global\long\def\indi#1{\chi\brq{#1}}

\global\long\def\evalat#1#2{ #1 \Big|_{#2}}

\global\long\def\Df#1#2{\frac{\partial#1}{\partial#2}}

\global\long\def\hx#1{\hat{x}_{#1}}

\global\long\def\eps{\varepsilon}

\global\long\def\M{{\cal M}}

\global\long\def\pr#1{\mathbb{P}\brq{#1}}

\global\long\def\opstyle#1{\mathbb{#1}}

\global\long\def\ex#1{\mathbb{E}\left[#1\right]}

\global\long\def\xp#1{\mathbb{E}\left[#1\right]}

\global\long\def\prcond#1#2{\opstyle P \left[\left. #1\ \right\vert \ #2 \right]}

\global\long\def\excond#1#2{\opstyle E \left[#1 \left|#2\right. \right]}

\global\long\def\exls#1#2{\opstyle{\opstyle E}_{#1}\left[ #2 \right]}

\global\long\def\prls#1#2{\opstyle P_{#1}\left[ #2 \right]}

\global\long\def\br#1{\left( #1 \right)}

\global\long\def\brq#1{\left[ #1 \right]}

\global\long\def\brw#1{\left\{  #1\right\}  }

\global\long\def\size#1{\left|#1\right|}

\global\long\def\set#1{\left\{  #1\right\}  }

\global\long\def\adj#1{\delta\br{#1}}

\global\long\def\st#1{[#1] }

\global\long\def\rin#1{r^{in}\br{#1}}

\global\long\def\rout#1{r^{out}\br{#1}}

\global\long\def\rini#1#2{r_{#1}^{in}\br{#2}}

\global\long\def\routi#1#2{r_{#1}^{out}\br{#2}}

\global\long\def\Iin{{\cal I}^{in}}

\global\long\def\Iout{{\cal I}^{out}}

\global\long\def\kin{k_{in}}

\global\long\def\kout{k_{out}}

\global\long\def\Min#1{{\cal M}^{in\brq{#1}}}

\global\long\def\Mout#1{{\cal M}^{out\brq{#1}}}

\global\long\def\M{{\cal M}}

\global\long\def\deltaM{\delta^{\M}}

\global\long\def\betin#1#2{\beta_{#1}^{in\brq{#2}}}

\global\long\def\betout#1#2{\beta_{#1}^{out\brq{#2}}}

\global\long\def\Bgen#1#2#3{B_{#3\brq{#2}}\br{#1}}

\global\long\def\Btgen#1#2#3#4{B_{#4\brq{#2}}^{#3}\br{#1}}

\global\long\def\Bin#1#2{\Bgen{#1}{#2}{in}}

\global\long\def\Bout#1#2{\Bgen{#1}{#2}{out}}

\global\long\def\Btin#1#2#3{\Btgen{#1}{#2}{#3}{in}}

\global\long\def\Btout#1#2#3{\Btgen{#1}{#2}{#3}{out}}

\global\long\def\B#1#2{B_{#1}^{#2}}

\global\long\def\R{R\br x}

\global\long\def\indi#1{\chi\brq{#1}}

\global\long\def\chr#1{\mathbf{1}_{#1}}

\global\long\def\pxE{\br{p_{e}x_{e}}_{e\in E}}

\global\long\def\pxEt#1{\br{p_{e}x_{e}}_{e\in E^{#1}}}

\global\long\def\xE{\br{x_{e}}_{e\in E}}

\global\long\def\x#1{x^{#1}}

\global\long\def\xt#1#2{x_{#1}^{#2}}

\global\long\def\z#1{z^{#1}}

\global\long\def\S#1{S^{#1}}

\global\long\def\xEt#1{\br{x_{e}}_{e\in E^{#1}}}

\global\long\def\e{\bar{e}}

\global\long\def\evalat#1#2{ #1 \Big|_{#2}}

\global\long\def\Df#1#2{\frac{\partial#1}{\partial#2}}

\global\long\def\event#1{\mathcal{E}_{#1}^{out}}

\global\long\def\eventsucc#1{\mathcal{E}_{#1}^{in}}

\global\long\def\X#1{\hat{X}_{#1}}

\global\long\def\E{\hat{E}}

\global\long\def\block#1{\Gamma\br{#1}}

\global\long\def\outblock#1{\Gamma^{out}\br{#1}}

\global\long\def\inblock#1{\Gamma^{in}\br{#1}}

\global\long\def\jblock#1{\Gamma^{j}\br{#1}}

\global\long\def\Frac#1#2{#1\left/\br{#2}\right.}

\global\long\def\P#1{{\cal P}\br{#1}}

\global\long\def\I{{\cal I}}

\global\long\def\basephi#1#2#3{\phi\brq{#1,#2}\br{#3}}

\global\long\def\setphi#1#2#3{\phi\brq{#1,#2}\br{#3}}

\global\long\def\support#1{\mbox{support}\br{#1}}

\global\long\def\act#1{act\br{#1}}

\global\long\def\myatop#1#2{\genfrac{}{}{0pt}{}{#1}{#2}} 

\global\long\def\myatop#1#2{\genfrac{}{}{0pt}{}{#1}{#2}} 

\global\long\def\evx#1#2#3{\mathbf{X}_{#2,#3}^{#1}}

\global\long\def\lam{\lambda}

\newcommand{\initOneLiners}{%
    \setlength{\itemsep}{0pt}
    \setlength{\parsep }{0pt}
    \setlength{\topsep }{0pt}
}
\newenvironment{OneLiners}[1][\ensuremath{\bullet}]
    {\begin{list}
        {#1}
        {\initOneLiners}}
    {\end{list}}

\title{Random Order Contention Resolution Schemes \footnote{This is an extended version of a paper whose preliminary version appeared in \emph{Proceedings of 2018 IEEE 59th Annual Symposium on Foundations of Computer Science}}}

\author{Marek Adamczyk\footnote{\texttt{m.adamczyk@mimuw.edu.pl} (The author's work is part of a project TOTAL that has received funding from the European Research Council (ERC) under the European Unions Horizon 2020 research and innovation programme (grant agreement No 677651)} \\ University of Warsaw
\and Micha\l{} W\l odarczyk\footnote{\texttt{m.wlodarczyk@mimuw.edu.pl} (The author has been supported by the National Science Centre of Poland Grant UMO-2016/21/N/ST6/00968.)} \\ University of Warsaw}

\date{}
\maketitle
\thispagestyle{empty}

\begin{abstract}
Contention resolution schemes
have proven to be an incredibly powerful concept which allows to tackle
a broad class of problems. The framework has been initially
designed to handle submodular optimization under various
types of constraints, that is, intersections of
exchange systems (including matroids), knapsacks, and unsplittable flows on trees.
Later on, it turned out that this framework perfectly extends to optimization
under uncertainty, like stochastic probing
and online selection problems, which further can be applied to mechanism
design.

We add to this line of work by showing
how to create contention resolution schemes for intersection of matroids and knapsacks when we work in the random order setting. More precisely,
we do know the whole universe of elements in advance, but they
appear in an order given by a random permutation. Upon arrival we need to irrevocably
decide whether to take an element or not.
We bring a novel technique for analyzing procedures in the random order setting that is based on the martingale theory.
This unified approach makes it easier to combine constraints, and we do not need to rely on the monotonicity of contention resolution schemes.

Our paper fills the gaps,
extends, and creates connections between many previous results and
techniques.
The main application of our framework is a $k+4+\eps$ approximation ratio for the Bayesian multi-parameter unit-demand mechanism design under the constraint of $k$ matroids intersection, which improves upon the previous bounds of $4k-2$ and $e(k+1)$.
Other results include improved approximation ratios for stochastic $k$-set packing and submodular stochastic probing over arbitrary non-negative submodular objective function, whereas previous results required the objective to be monotone.

\end{abstract}
\newpage

\setcounter{page}{1}

\tableofcontents

\newpage
\section{Introduction}
Uncertainty in input data is a common feature of most practical problems and
research in finding good solutions (both experimental and theoretical) for such
problems has a long history. In recent years one technique in particular has turned out to be very effective in tackling such problems, namely the Contention Resolution Schemes (CR schemes).
They have been introduced by Chekuri et al.~\cite{DBLP:conf/stoc/VondrakCZ11} in order to maximize submodular functions under various constraints.
Submodular functions have proven important in modeling various optimization problems
that share the property of diminishing returns.

This framework has been initially designed for problems in deterministic setup, where all information is known at the beginning. However, its randomized approach has turned out to be perfect to tackle problems where the uncertainty was the part of the model, like stochastic probing and mechanism design~\cite{DBLP:conf/ipco/GuptaN13}. 

This fact was elegantly leveraged by Feldman et al.~\cite{DBLP:conf/soda/FeldmanSZ16}, who adapted the framework of CR schemes to an online setting, and resolved a long-standing open question by Chawla et al.~\cite{DBLP:conf/stoc/ChawlaHMS10}, by devising a so called Oblivious Posted Price Mechanism for matroids.
This implied a constant factor approximations for the Bayesian multi-parameter unit-demand mechanism design problem. 

Inspired by this line of work we have asked ourselves a question:
\begin{quote}
\emph{What can contention resolution schemes do, if we shall consider them in the random order model?}
\end{quote}
While trying to answer this question we drew from, extended, bridged some gaps between, and improved some of the results on CR schemes~\cite{DBLP:conf/stoc/VondrakCZ11,DBLP:conf/soda/FeldmanSZ16}, sequential posted price and multi-parameter mechanism design~\cite{DBLP:conf/stoc/ChawlaHMS10,DBLP:conf/ipco/GuptaN13,DBLP:conf/stoc/KleinbergW12,DBLP:conf/soda/FeldmanSZ16}, and stochastic probing ~\cite{DBLP:conf/ipco/GuptaN13,DBLP:journals/mor/AdamczykSW16}.
We describe these results precisely below.

\subsection{Problems overview, known results, and our contributions}

\paragraph*{Contention resolution schemes}

Let us start with an illustrative problem.
Consider a matroid $\M = \br{E,\I}$ and a fractional solution $x$ from its polytope. Suppose we are given a weight vector $w: E\mapsto \mathbb{R}_+$, and we look for an algorithm that returns an independent set $S \in \I$ such that $\sum_{e\in S} w_e \geq  c\cdot \sum_{e\in E} w_e x_e$ for some constant $c < 1$.
The idea is to settle for a randomized algorithm 
and demand that every element is taken into $S$ with probability at least $c\cdot x_e$. Such a property would immediately entail the desired guarantee.

How to design an algorithm returning $S$ such that $\pr{e\in S} \geq c\cdot x_e$? 
Chekuri et al.~\cite{DBLP:conf/stoc/VondrakCZ11} presented a framework
of contention resolution schemes (CR schemes) which address this problem, among other applications. The idea is to first draw a random set $\R$ such that $\pr{e \in \R} = x_e$ for each $e \in E$ independently, and afterwards -- since $\R$ is most likely not an independent set in $\I$ -- to drop some elements from $\R$ to meet the feasibility constraint, that is, to resolve the contention between the elements.

\emph{\underline{Our contribution:}} 
Simply speaking, we show that the above problem can be solved also if we work in a random order model, i.e., when elements of $E$ appear to us according to a uniformly random  permutation, and upon arrival we need to make an irrevocable decision of whether to take an element or not.

In its full generality Chekuri et al.~were dealing not only with matroids but arbitrary intersections of matroids, knapsacks, exchange systems, and unsplittable flow on trees. They were also maximizing not only linear functions, but non-negative submodular functions as well.
We do so as well, restricted to intersections of matroid and knapsack constraints.
For a single matroid and a linear objective, Chekuri et al.~obtained an approximation (the  constant $c$) of $1-\frac{1}{e}$, while we get $\frac{1}{2}$. However, for intersection of $k$ matroids, starting with $k\geq 2$, we obtain a better bound of $\frac{1}{k+1}$, improving upon theirs $\frac{1}{e\cdot k + o(k)}$, even though we work in a more restrictive model.
\begin{restatable}{thm}{crscheme}
\label{thm:crscheme}
There exists a random-order CR scheme for intersection of $k$ matroids
with $c=\frac{1}{k+1}$.
\end{restatable}
A possible explanation for this -- for a moment we assume that the Reader is familiar with the previous work --
is that, unlike the previous CR schemes, we do not require the monotonicity
of the scheme. Monotonicity appeared to be an important feature
because it allowed to combine the schemes via the FKG inequality~\cite{Alon:2016:PM:3002498}.
We manage to combine the schemes for matroids, sparse column packings,
and knapsacks without the monotonicity requirement, and we believe
it is an interesting fact on its own.

For submodular objective we also improve the bounds starting with $k\geq 2$.
\begin{restatable}{thm}{submodular}
\label{thm:submodular-main}
Maximization of a non-negative submodular function with respect to $k$ matroid constraints admits a $(k+1+\eps)\cdot e$ approximation algorithm in the random-order model.
\end{restatable}

These results are not absolutely best when compared to more general techniques,
since one can get ratio $(k-1)$ for linear objectives when $k\ge 2$ using iterative
rounding~\cite{Lau:2011:IMC:2018773}, and $(k+2)$ for non-negative submodular
functions via a combinatorial argument~\cite{DBLP:journals/siamdm/LeeMNS10}.
However, to the best of our knowledge, our results yield the best ratio in the random order model.

\paragraph{Mechanism Design}

Consider the following mechanism design problem. There are $n$ agents
and a single seller providing a set of services.
The agent $i$ is interested in buying the $i$-th service and values
its as $v_{i}$, which is drawn independently from a distribution
$D_{i}$.
Such a setting is called single-parameter.
The valuation $v_{i}$ is
private, but the distribution $D_{i}$ is known in advance. The seller can provide
only a subset of services, that belongs to a system ${\cal I}\in2^{\brq n}$,
which is specified by feasibility constraints. A mechanism accepts bids
of agents, decides on subset of agents to serve, and sets individual
prices for the service. A mechanism is called truthful if agents are motivated to bid
their true valuations. Myerson's theory of virtual valuations yields
truthful mechanisms that maximize the expected revenue of a
seller~\cite{Myerson:1981:OAD:2781650.2781656}, although they sometimes might be impractical~\cite{Ausubel06thelovely}. On the other
hand, practical mechanisms are often non-truthful~\cite{Ausubel06thelovely}. The Sequential
Posted Pricing Mechanism (SPM) introduced by Chawla~et~al.~\cite{DBLP:conf/stoc/ChawlaHMS10}
gives a nice trade-off -- it is truthful, simple to implement, and
gives near-optimal revenue. An SPM offers each agent a 'take-it-or-leave-it'
price for a service. After refusal the service shall not be provided,
so it is easy to see that an SPM is indeed a truthful mechanism.

The paragraph above concerns only the single-parameter setup. In the Bayesian multi-parameter
unit-demand mechanism design (BMUMD for short), we have $n$ buyers
and one seller. The seller offers a number of different services indexed
by set ${\cal J}$. The set ${\cal J}$ is partitioned into groups
${\cal J}_{i}$, with the services in ${\cal J}_{i}$ being targeted
by agent $i$. Each agent $i$ is interested in getting any one of
the services in ${\cal J}_{i}$, i.e., agents are unit-demand. Agent
$i$ has value $v_{j}$ for service $j\in{\cal J}_{i}$. Value $v_{j}$
is independent of all other values and is drawn from distribution
$D_{j}$. Once again the seller faces a feasibility constraint specified
by a set system~${\cal I} \subseteq2^{J}$.

Unlike single-parameter setup, this problem is not solvable efficiently
by the well-established Myerson's approach. The paper of Chawla et al.~\cite{DBLP:conf/stoc/ChawlaHMS10}
launched a line of work in obtaining approximate results
for the multi-parameter setup, by suggesting a possible avenue of
a solution via the so-called Oblivious Posted Price mechanisms. One
would have to first embed the multi-parameter problem into a single-parameter
one, and later to ensure that the algorithm would work if the items
are presented in an adversary order.
Kleinberg
and Weinberg~\cite{DBLP:conf/stoc/KleinbergW12} solved the BMUMD problem for matroid environments
with approximation of $4k-2$ for intersection of $k$ matroids (with
2-approximation for a single matroid), but they have not used the Oblivious
Posted Price mechanisms. Feldman et al.~\cite{DBLP:conf/soda/FeldmanSZ16} devised the first Oblivious
Posted Price mechanisms and obtained an $e k + o(k)$
approximation for the intersection of $k$ matroids.

\emph{\underline{Our contribution:}} 
We observe that the Oblivious
Posted Price is an overly demanding notion, and we need to handle
the oblivious order only when looking at the items of a given client,
but there is no need to restrict the order of
clients. In our algorithm we randomly shuffle clients, but
cannot make assumption on the client's choice.
This hybrid approach is what allows us to
obtain improved bounds.
For $k=2$ we match up to $\eps$ the 6-approximation of Kleinberg and Weinberg~\cite{DBLP:conf/stoc/KleinbergW12}, but starting from $k\geq 3$ our ratios are better; for $k=3$ we get $7+\eps$ improving over $9.48$ of Feldman et al.~\cite{DBLP:conf/soda/FeldmanSZ16}.
\begin{restatable}{thm}{bmumd}
\label{thm:bmumd-main}
Bayesian multi-parameter unit-demand mechanism design over $k$ matroid constraints admits a $\br{k+4+\eps}$ approximation for any $\eps > 0$.
\end{restatable}

\paragraph{Non-negative submodular stochastic probing}

We are given a universe $E$, where each element $e\in E$ is \emph{active
}with probability $p_{e}$ independently. The only way to
find out if an element is active is to \emph{probe }it. We call a
probe \emph{successful} if an element turns out to be active. 
We execute an algorithm that probes the elements one-by-one.
If an element is active, the algorithm is forced to add it to the
current solution. In this way, the algorithm gradually constructs
a solution consisting of active elements.

We consider the case in which we are given constraints on both the
set of probed elements and the set of elements included in the solution. Formally,
we are given two downward-closed independence systems: an \emph{outer} system $\br{E,\Iout}$ restricting
the set of elements probed by the algorithm, and an \emph{inner}
system $\br{E,\Iin}$, restricting the set of elements taken by the
algorithm. The goal is to maximize the expected value
$\ex{f\br S}$, where $f$ is a given non-negative submodular function and $S$
is the set of all successfully probed elements.

This problem has been stated by Gupta and Nagarajan~\cite{DBLP:conf/ipco/GuptaN13}
who gave an abstraction for couple of problems like stochastic matching
and sequential-posted price mechanisms in a single-parameter setup.
They obtained an $O(\kin + \kout)$ approximation for linear objectives in an environment with $\kin$ inner matroids and $\kout$
outer matroids (together with results for more general constraints) using the CR-schemes of
Chekuri et al.~\cite{DBLP:conf/stoc/VondrakCZ11}. Later, Adamczyk et al.~\cite{DBLP:journals/mor/AdamczykSW16} showed how
to obtain a $\br{\kin+\kout}$-approximation for linear objectives
and $\frac{e}{e-1}\cdot\br{\kin+\kout+1}$ for monotone submodular objectives.

\emph{\underline{Our contribution:}} We obtain the first results with
respect to arbitrary non-negative submodular objective functions.
\begin{restatable}{thm}{probing}
\label{thm:probing-main}
Non-negative submodular stochastic probing with $k_{in}$ inner matroid constraints and $k_{out}$ outer matroid constraints admits a $(k_{in}+k_{out}+1+\eps) \cdot e$ approximation for any $\eps>0$. 
\end{restatable}

\paragraph{Stochastic $k$-set packing}

We are given $n$ elements/columns, where each element $e\in\brq n$
has a random profit $v_{e}\in\mathbb{R}_{+}$, and a random $d$-dimensional
size $L_{e}\in\{0,1\}^{d}$. The sizes are independent for different
elements, but $v_{e}$ can be correlated with $L_{e}$, and the coordinates
of $L_{e}$ also might be correlated between each other.
The values of $v_e$ and $L_e$ are revealed after $e$ is probed, but their distributions are known in advance.

Additionally,
for each element $e$ we are given a set $Q_{e}\subseteq[d]$ of size at most
$k$, such that the size vector $L_{e}$ takes positive
values only in these coordinates, i.e., $L_{e}\subseteq Q_{e}$ with
probability 1. We are also given a~capacity vector $b\in\mathbb{Z}_{+}^{d}$
into which elements must be packed,
that is, the solution can consist of at most $b_i$ elements with unit sizes in the $i$-th row.
We say that the outcomes of $L_{e}$
are monotone if for any possible realizations $x,y\in\set{0,1}^{d}$
of $L_{e}$, we have $x\leq y$ or $y\leq x$ coordinate-wise.

A~strategy
probes columns one by one, obeying the packing constraints, and the
goal is to maximize the expected outcome of taken columns. The stochastic $k$-set packing problem was stated by Bansal
et al.~\cite{Bansal:woes}. They have presented a $2k$-approximation algorithm for
it, and a $\br{k+1}$-approximation algorithm with an assumption
that the outcomes of size vectors $L_{e}$ are monotone. Recently
Brubach et al.~\cite{DBLP:conf/soda/BrubachSSX18} improved the approximation ratio to $k+o(k)$ in
the general case.

\emph{\underline{Our contribution:}} We improve upon
the recent bound of Brubach et al.~\cite{DBLP:conf/soda/BrubachSSX18}.
Our algorithm also works in the case where we replace counting constraints on rows with arbitrary matroids.
\begin{restatable}{thm}{packing}
\label{thm:ksetpacking}There exists a $\br{k+1}$ approximation algorithm
for stochastic $k$-set packing over matroid row constraints. 
\end{restatable}

\subsection{Our techniques}

The main notion we use is a \emph{controller mechanism}, which provides a handy abstraction,
that allows us to combine various constraints without relying on the monotonicity of the schemes.
For matroids it is implemented using
a decomposition of a fractional solution into a convex combination of characteristic vectors of independent
sets, and for knapsacks a controller is represented as a point from
the unit interval.
Additionally, knapsack constraints require a preprocessing procedure, that partitions the elements into big and small, which is inspired by~\cite{DBLP:conf/ipco/BansalKNS10,DBLP:conf/soda/FeldmanSZ16}.

The controller mechanism of a constraint $\mathcal{I}$
randomly assigns each element $e\in E$ a~controller $C_{e}$, which
keeps track of its suitability to become a part of the solution when
we iterate through the elements in a random order. More formally, 
\begin{OneLiners}
\item[a)] if $S$ is the current solution and $C_{e}$ has not been \emph{blocked} yet,
then $S\cup\{e\}$ must belong to $\mathcal{I}$, 
\item[b)] for each element $e$ the probability that 1) some element $f$ has been chosen at step $t$, and 2) $f$ has been
assigned a controller $C_{f}$, that blocks $C_{e}$, is at most $\frac{\lambda}{n-t}$ (probability taken over all such $f$'s and $C_f$'s),
for a constant $\lambda$ depending on $\mathcal{I}$. 
\end{OneLiners}

With these properties on hand, we can associate a submartingale with each element
$e$ and a fixed controller~$C_{e}$. We define a stopping event of
revealing the fate of $e$, i.e., we stop when we either take $e$ into the solution
or we block its controller. Before the stopping event for $e$ occurs, we know that we still can either take it or block it. The bound on the probability
of accepting the element comes then from the Doob's stopping theorem. This suffices
to construct a random-order contention resolution scheme. Another
martingale argument extends this reasoning to the submodular function
maximization.

In the context of the stochastic probing problems, we are aware of only one usage of the martingale argument with the Doob's theorem, in the analysis of an iterative randomized rounding algorithm~\cite{DBLP:journals/mor/AdamczykSW16}.
To the best of our knowledge, we present the first application of the martingale argument to analyze a random permutation, and we believe this technique can be handy and worth adding to a toolbox.

In order to handle Bayesian multi-parameter unit-demand mechanism
design, we rely on the reduction to a single-parameter setup by Chawla et al.~\cite{DBLP:conf/stoc/ChawlaHMS10} via copies, and on the linear relaxation by Gupta and Nagarajan~\cite{DBLP:conf/ipco/GuptaN13}.
The last ingredient necessary to obtain the postulated approximation ratio
for $k$ matroids is a routine that processes a fractional solution
for a single client menu, which later on allows to give very tight upper and lower bounds on the probabilities
of an item's acceptance and rejection. We present such a routine based
on local search that reduces the discrepancy between these quantities
in each step.

Arguments for stochastic probing and stochastic $k$-set packing
exploit the same notion of the controller mechanism.
However, in order to obtain an upper bound for a submodular objective case we need a stronger guarantee for the measured continuous greedy algorithm for optimizing the multilinear extension of a submodular function~\cite{DBLP:conf/focs/FeldmanNS11}. This bound is due to Justin Ward~\cite{Justin15private}.

\subsection{Organization of the paper}

We start the technical part of the paper by showing a random-order CR scheme for a matroid in Section~\ref{sec:matroid-cr}.
Section~\ref{sec:controller} contains the analysis of the CR scheme
and introduces the language of our framework, that is, the controller mechanism and characteristic sequences.
This allows us to  present the extension to multiple matroids in a simple way, and later to explain how to deal with submodular functions.

In Section~\ref{sec:bmumd}
we present the more complicated algorithm for the Bayesian multi-parameter unit-demand mechanism design.
The details of the single-client routine are postponed to Section~\ref{sec:single-client}.
This order of presentation allows us to explain both the framework and the main result relatively soon.
The following Sections~\ref{sec:submodular},~\ref{sec:packing},
and~\ref{sec:probing} cover the submodular optimization, stochastic $k$-set packing, and stochastic probing.

We deliberately avoid giving one procedure that captures all the results at once for the cleanest possible presentation of the paper.
With each result comes an abstract formulation of the algorithm and the application in the matroid environment. 
Our framework also extends to knapsack constraints,
and we show how to combine them with matroids in Section~\ref{sec:knapsack}.

\section{Preliminaries}

\subsection{Submodular functions}

A set function $f:2^{E}\mapsto\mathbb{R}_{\ge0}$ is \emph{submodular},
if for any two subsets $S,T\subseteq E$ we have $f\br{S\cup T}+f\br{S\cap T}\leq f\br S+f\br T$.
The \emph{multilinear extension} of $f$ is a function $F:[0,1]^{E}\mapsto\mathbb{R}_{\ge0}$, whose value at a point $y\in\brq{0,1}^{E}$ is given by 
\[
F\br y=\sum_{A\subseteq E}f\br A\cdot\prod_{e\in A}y_{e}\prod_{e\not\in A}\br{1-y_{e}}.
\]
Note that $F\br{\chr A}=f\br A$ for any set $A\subseteq E$, so $F$
is an extension of $f$ from discrete domain $2^{E}$ into a real
domain $\brq{0,1}^{E}$. The value $F(y)$ can be interpreted as the
expected value of $f$ on a random subset $A\subseteq E$ that is
constructed by taking each element $e\in E$ with probability $y_{e}$.

\subsection{Matroids }

For a matroid $\M=\br{E,\I\subseteq2^{E}}$, we define its \emph{ matroid polytope} $${\cal P}\br{{\cal M}}=\setst{x\in\mathbb{R}_{\geq0}^{E}}{\forall_{A\in{\cal I}}\sum_{e\in A}x_{e}\leq r_{{\cal M}}\br A},$$
where $r_{{\cal M}}$ is the rank function of ${\cal M}$. We know
that ${\cal P}\br{{\cal M}}$ is equivalent to
the convex hull of $\setst{\chr A}{A\in{\cal I}}$, i.e.\ characteristic
vectors of all independent sets of ${\cal M}$.

We shall need the following two properties. The proof of the lemma below about the existence of a convex decomposition can be found in~\cite{Schrijver:book}.
\begin{lem}
\label{prelim:support}
We can represent any $x\in{\cal P}\br{{\cal M}}$ as $x=\sum_{i=1}^{m}\beta_{i}\cdot\chr{B_{i}}$,
where $B_{1},\ldots,B_{m}\in{\cal M}$ and $\beta_{1},\ldots,\beta_{m}$
are non-negative weights such that $\sum_{i=1}^{m}\beta_{i}=1$
and $m = |E|^{O(1)}$.
We denote ${\cal S} = [m]$ and call $\left({\cal S}, (B_i)_{i \in \cal{S}}, (\beta_i)_{i \in \cal{S}}\right)$
a~\emph{support }of $x$
in ${\cal P}\br{{\cal M}}$. 
\end{lem}

The following lemma is a slightly generalized basis exchange lemma, proof of which again can be found in~\cite{Schrijver:book}.

\begin{lem}
\label{prelim:mapping}
Let $A,B\in{\cal I}$ be two independent sets of matroid ${\cal M}=\br{E,{\cal I}}$.
We can find an exchange-mapping $\phi\brq{A,B}:A\mapsto B\cup\set{\bot}$
such that:

\begin{enumerate}
\item $\basephi ABe=e$ for every $e\in A\cap B$, 
\item for each $f\in B$ there exists at most one $e\in A$ for which $\basephi ABe=f$, 
\item for $e\in A\setminus B$, if $\basephi ABe=\bot$, then $B+e\in{\cal I}$,
otherwise $B-\basephi ABe+e\in{\cal I}$. 
\end{enumerate}
\end{lem}

\subsection{Martingales}

\begin{defn} Let $\left(\Omega,{\cal F,\mathbb{P}}\right)$
be a probability space, where $\Omega$ is a sample space, ${\cal F}$
is a $\sigma$-algebra on $\Omega$, and $\mathbb{P}$ is a probability
measure on $(\Omega,{\cal F)}$. Sequence $\left\{ {\cal F}_{t}:t=1,2,\dots\right\} $
is called a \emph{filtration} if it is an increasing family of sub-$\sigma$-algebras
of ${\cal F}$: ${\cal F}_{0}\subseteq{\cal F}_{1}\subseteq\ldots\subseteq{\cal F}$.
\end{defn} Intuitively speaking, when considering a stochastic
process, $\sigma$-algebra ${\cal F}_{t}$ represents all information
available to us right after making step $t$. In our case $\sigma$-algebra
${\cal F}_{t}$ contains all information about each randomly chosen
element to probe, about outcome of each probe, and about each controller
update, that happened before or at step $t$. \begin{defn}
A process $\left(Z_{t}\right)_{t=1}^n$ is called a \emph{martingale}
if for every $t\geq0$ all following conditions hold:
\begin{enumerate}
\item random variable $Z_{t}$ is ${\cal F}_{t}$-measurable,
\item $\ex{\left|Z_{t}\right|}<\infty$,
\item $\excond{Z_{t+1}}{\mathcal{F}_{t}}=Z_{t}$. 
\end{enumerate}
If we replace the latter condition with $\excond{Z_{t+1}}{\mathcal{F}_{t}} \ge Z_{t}$,
we obtain a \emph{submartingale}.
\end{defn}

\begin{defn} Random variable $\tau:\Omega\mapsto\left\{ 0,1,\ldots\right\} $
is called a \emph{stopping time} if $\left\{ \tau=t\right\} \in\mathcal{F}_{t}$
for every $t\geq0$. \end{defn} Intuitively, $\tau$ represents
a moment when a particular event happens. We have to be able to say whether
it happened at step $t$ given only the information from steps $0,1,2,\ldots,t$.
In our case we define $\tau$ as the moment when we get to know the fate of an element,
i.e., either when it was selected in a given step, or when its blocking event occurred.
It is clear that this is a stopping time according to the above definition.

\begin{thm} [Doob's Optional-Stopping Theorem] Let $\left(Z_{t}\right)_{t=1}^n$
be a submartingale. Let $\tau$ be a stopping time such that $\tau$
has finite expectation, i.e., $\mathbb{E}[\tau]<\infty$, and the
conditional expectations of the absolute value of the martingale increments
are bounded, i.e., there exists a constant $c$ such that $\mathbb{E}\bigl[|Z_{t+1}-Z_{t}|\,\big\vert\,\mathcal{F}_{t}\bigr]\le c$
for all $t\geq0$. If so, then $\ex{Z_{\tau}}\ge\ex{Z_{0}}$. \end{thm}

\section{Random-order contention resolution scheme for a matroid}
\label{sec:matroid-cr}

We formulate our first goal as a motivation to present the simplest variant of the mechanism.
\begin{thm}
\label{thm:matroid-cr}
There exists a random-order CR scheme for a  matroid
with $c=\frac{1}{2}$.
\end{thm}

\paragraph{Initialization}
The procedure is shown in Algorithm~\ref{alg:matroid-cr}.
Given a vector $x\in\P\M$, we begin with decomposing it into a~\emph{support} $x = \sum_{j \in\cal S}\beta_{j}\cdot \chr{B_{j}^{0}}$, where each set $B_i^0$ is independent in $\cal M$ (Lemma~\ref{prelim:support}), and finding exchange-mappings $\phi\brq{B_{i}^{0},B_{j}^{0}}$ between each pair of sets in the support (Lemma~\ref{prelim:mapping}).
For each element $e \in E$ we choose a~controller $j(e) \in\cal S$ such that $e\in B_{j}^0$, with
probability $\frac{\beta_{j}}{x_{e}}$
(note that $\sum_{j : e \in B_j^0} \beta_j = x_e$).
The set family given by the support is being modified after each step of the algorithm
and we denote the sets in step $t$ as $(B_j^t)_{j\in\cal S}$.
The set $\cal S$ and scalars $\beta_j$ remain the same.
For the sake of legibility we refer directly to set $B_{j(e)}^t$ as $C_e^t$
and shorten it to $C_e$ when it does not lead to a confusion.

\paragraph{Blocking events}
We scan elements from $E$ in a random order.
If the element $e$ chosen in step $t$ happens to belong to $R(x)$ and its controller has not been blocked yet (to be explained shortly), we take it into the solution.
Then we modify the set family family $(B_j^t)_{j\in\cal S}$ by inserting $e$ to each of them.
This operation is performed according to the exchange-mappings.
It may result in some other element $f$ being removed from the set $C_f^t = B_{j(f)}^t$.
When this happens, we say that $(f,C_f)$ gets \textbf{blocked}.

We emphasize that at the moment of doing so, in some circumstances,
it would be still possible to take element $f$ into the solution.
However, we require a clean condition to know when an element is not considered any longer. This
simplifies the analysis significantly. 
In the pseudocode shown below, we check for the blocking event of
$e$ in line~\ref{algline:majorif}.

\begin{algorithm}
\caption{\label{alg:matroid-cr}Random-order contention resolution scheme for a matroid}

\begin{algorithmic}[1]


\STATE decompose $x$ into its support in $\cal M$, that is, $x = \sum_{i \in\cal S}\beta_{i}\cdot \chr{B_{i}^{0}}$

\STATE find exchange-mappings $\phi\brq{B_{i}^{0},B_{j}^{0}}$
between all pairs $i,j\in\cal S$

\STATE for each element $e$ choose a controller $j(e) \in\cal S$ such that $e\in B_{j}^0$ with
probability $\frac{\beta_{i}}{x_{e}}$,
denote $B_{j(e)}^0$ by $C_{e}^0$ \label{algline:controller}

\STATE $S\leftarrow\emptyset$, $t\leftarrow0$


\STATE \textbf{for} each element $e$ in $E$ in $\sigma$ order \textbf{do}

\STATE ~~~\textbf{if $e\notin R(x)$ then \label{algline:Roracle}}

\STATE ~~~~~~continue

\STATE ~~~\textbf{if $e\in C_{e}^t$ then \label{algline:majorif}}

\STATE ~~~\textbf{~~~}$S\leftarrow S\cup\set e$
\label{algline:adde}


\STATE ~~~\textbf{~~}~\textbf{for }each $i\in{\cal S}:e\notin B_{i}^{t}$
\textbf{do}

\STATE ~~~~~~~~~\textbf{if }$\phi\brq{C_{e}^t, B_{i}^{t}}\br e=\perp$
\textbf{then} $B_{i}^{t+1}\leftarrow B_{i}^{t}+e$

\STATE ~~~~~~~~~\textbf{if }$\phi\brq{C_{e}^{t},B_{i}^{t}}\br e=f$
\textbf{then} $B_{i}^{t+1}\leftarrow B_{i}^{t}-f+e$\label{algline:addetoB}

\STATE ~~~\textbf{~}~~\textbf{for }each $i\in{\cal S}:e\in B_{i}^{t}$
\textbf{do}

\STATE ~~~~~~\textbf{~~$B_{i}^{t+1}\leftarrow B_{i}^{t}$}

\STATE ~~~~~~find new exchange-mappings $\phi\brq{B_{i}^{t+1},B_{j}^{t+1}}$
between all pairs $i,j\in\cal S$

\STATE ~~~$t\leftarrow t+1$; 

\RETURN $S$

\end{algorithmic} 
\end{algorithm}

\paragraph{Correctness}

Let $S^t$ stand for the solution constructed up to step $t$. We need to show that the output is indeed an independent set
of the matroid.
This follows from the two facts below.
\begin{fact}
For every $t$ and $i \in\cal S$ it holds $S^{t}\subseteq B_{i}^{t}$.
\end{fact}

\begin{proof}
If we add an element $e$ to $S^{t}$ on line~\ref{algline:adde},
then we add $e$ to each $B_{i}^{t}$.
\end{proof}
\begin{fact}
For every $t$ and $i \in\cal S$ the set $B_{i}^{t}$
is independent in the matroid $\M$.
\end{fact}

\begin{proof}
All changes of the sets $B_{i}^{t}$ are due do the exchange-mapping $\phi$
whose property (3) ensures that after each exchange sets
$B_{i}^{t}$ remain independent in $\M$.
See Lemma~\ref{prelim:mapping} for details.
\end{proof}

\paragraph{Approximation guarantee}
In our setting
we cannot assume we know the whole set $R\br x$ in advance,
but rather we learn if $e \in R(x)$ after probing $e$ in line~\ref{algline:Roracle}.
In the following arguments we fix an element $e$ and condition
all the probabilities on the fact that $e\in R\br x$, and on the
controller $C_{e}$ chosen in line~\ref{algline:controller}.
Since the choice of other controllers is irrelevant to $e$ until an element $f$ with a controller blocking $C_e$ is revealed to exist in line~\ref{algline:Roracle},
we can assume in the analysis that the assignment of $C_f$ happens after the latest family of exchange-mappings has been established.

The next two lemmas encapsulate the properties of the controller mechanism for a matroid.
The main proof is postponed to Lemma~\ref{lem:cr}.

\begin{lem}
\label{lem:matroid:support}
Suppose that $\sum_{j\in\cal S}\beta_j \le 1$ and $(B_j)_{j \in \cal{S}}$ is a family of independent sets from $\cal M$
with fixed exchange-mappings between each $B_j$ and set $C \in \cal M$.
Let us denote by $\Gamma\br{e,C}=\setst{\br{f,j}}{\,\phi_{[B_{j},C]}\br f=e}$
the set of all pairs $(f,j)$ that makes $e$ get removed from $C$.
Then
\[
\sum_{f \in E}\,\,\sum_{j:\br{f,j}\in\Gamma\br{e,C}}\beta_{j}\le 1.
\]
\end{lem}
\begin{proof}
For every set $B_{j}$ there can be at
most one element $f$ such that $\br{f,j}\in\Gamma\br{e,C}$ because
$\phi\brq{B_{j},C}$ cannot map two elements onto $e$ (Lemma~\ref{prelim:mapping}).
Therefore for fixed $j\in \cal S$ we have $\sum_{f:\br{f,j}\in\Gamma_{i}\br{e,C}}\beta_{j}\leq\beta_{j}$.
We change the summation order to obtain
\[
\sum_{f \in E} \,\,\sum_{j:\br{f,j}\in\Gamma\br e,C}\beta_{j}=\sum_{j \in \cal S}\,\,\sum_{f:\br{f,j}\in\Gamma\br{e,C}}\beta_{j}\leq\sum_{j\in\cal S}\beta_{j}\le 1.
\]
\end{proof}

\begin{lem}
\label{lem:matroid-blocking-pr}
The probability of a blocking event for $(e, C_e)$ in step $t$ is at most $\frac{1}{n-t}$.
\end{lem}
\begin{proof}
We enumerate steps starting with 0. A~blocking event occurs when we remove $e$ from $C_{e}^t$.
This happens if we choose $f\neq e$ in step $t$, that 1) turns out to belong to $R\br x$ in line~\ref{algline:Roracle}, and
2) we choose a controller $C_f$ such that $\phi_{C_f^{t},C_{e}^{t}}\br f=e$ in line~\ref{algline:controller}
(recall that in our analysis we can treat this event as happening after the existence $f$ has been revealed).
Let $\Gamma^{t}\br {e,C_e}$ be as in Lemma~\ref{lem:matroid:support} with respect to the set family $(B_j^t)_{j\in\cal S}$.
Since there are $n-t$ elements to choose in step $t$,
the probability that $e$ gets removed from $C_{e}^t$ is
at most 
\[
\frac{1}{n-t}\sum_{f}\sum_{j:\br{f,j}\in\Gamma^{t}\br{e,C_e}}\pr{f\in R\br x}\cdot\pr{f\mbox{ chooses controller }j}.
\]
We have $\pr{f\in R\br x}=x_{f}$.
If $f$ belongs to $R\br x$, then $f$
is assigned $C_f = B_{j}^{t}$ with probability $\frac{\beta_{j}}{x_{f}}$.
Therefore the above expression simplifies to 
\[
\frac{1}{n-t}\sum_{\br{f,j}\in\Gamma^{t}\br{e,C_e}}x_{f}\cdot \frac{\beta_{j}}{x_{f}}=\frac{1}{n-t}\sum_{\br{f,j}\in\Gamma^{t}\br{e,C_e}}\beta_{j}.
\]
The claim follows from Lemma~\ref{lem:matroid:support}.

\end{proof}

\section{The controller mechanism}
\label{sec:controller}
Before we are ready to finish the proof of Theorem~\ref{thm:matroid-cr}, we need to introduce our toolbox.
In this section we abstract from the structure of the constraint and present the general framework for obtaining approximation ratios with the controller mechanism.

\begin{algorithm}
\caption{\label{alg:abstract-cr}Abstract view of the random-order contention resolution scheme}
\begin{algorithmic}[1]

\STATE assign each element $e \in E$ a controller $C_e$

\STATE $S\leftarrow\emptyset$

\STATE \textbf{foreach} element $e$ in $E$ in $\sigma$ order \textbf{do}\label{algline:abstact-cr-pick}

\STATE ~~~\textbf{if $e\notin R(x)$ then} 

\STATE ~~~~~~continue

\STATE ~~~\textbf{if} $(e,\, C_{e})$ has not been blocked \textbf{then}  

\STATE ~~~\textbf{~~~}$S\leftarrow S\cup\set e$

\STATE ~~~\textbf{~~~}update controllers

\RETURN $S$
\end{algorithmic} 
\end{algorithm}

\subsection{Characteristic sequences}
\label{sec:character}
In order to analyze the approximation guarantee we fix an element $e$ and condition
all the probabilities on the fact that $e\in R\br x$, and on the choice of
controller $C_{e}$ (using notation $\pr{\mbox{event}\,|\,C_e}$).
The element $e$ is oblivious to the choice of other controllers until an element $f$ with a controller blocking $C_e$ is taken into the solution.
Hence, we can assume in the analysis that for $f\ne e$ the assignment of $C_f$ happens after the last controller update and the disclosure of $f$.

Initially we know that $e$ is \emph{available} to take, i.e., there
is still a possibility of accepting $e$ via $C_{e}^{t}$ for some $t$.
As the process is being executed, at some point we get to
know the fate of $e$: there comes a step in which we
either
1) pick $e$ in line~\ref{algline:abstact-cr-pick}, or
2) pick $f\neq e$ and choose a controller $C_f$
which blocks $(e,\,C_e)$ (in the matroid example: $e$ gets removed from $C_{e}$).

\begin{defn}[\textbf{Characteristic sequences}]
Consider an abstract routine, where in every turn each unseen element might be picked with equal probability and, if its controller has not been blocked, it gets accepted and might block other controllers.
We shall associate three binary processes with $(e,C_e)$:
\begin{OneLiners}
\item[$S_e^t:$] indicates whether $e$ was taken into the solution before step $t$;
initially $S_{e}^{0}=0$;
\item[$Z_e^t:$] indicates whether the controller of $e$ has been blocked before step $t$;
initially $Z_{e}^{0}=0$;
\item[$Y_e^t:$] we still \textbf{didn't get} to know the fate of $e$ before
step $t$; initially $Y_{e}^{0}=1$.
\end{OneLiners}
The sequences are bound with a following relationship
\[
Y_{e}^{t}=1-S_{e}^{t}-Z_{e}^{t}.
\]
We call the characteristic sequences $\lambda$-\textbf{bounded} if
\[
\excond{Z_{e}^{t+1}-Z_{e}^{t}}{{\cal F}^{t}}\leq\frac{\lambda\cdot Y_{e}^{t}}{n-t}.
\]
\end{defn}


\begin{cor}
\label{lem:matroid-lam}
For the matroid constraint the characteristic sequences are 1-bounded.
\end{cor}
\begin{proof}
First let us note that if $Y_{e}^{t}=0$, then we already got to know
the fate of $e$ before step $t$, and so the status of blocking
$e$ cannot change, i.e., $Z_{e}^{t}=Z_{e}^{t+1}$.
If $Y_{e}^{t}=1$, then the claim reduces to Lemma~\ref{lem:matroid-blocking-pr}.
\end{proof}

\begin{lem}
\label{lem:deltas}
If the characteristic sequences of $e$ are $\lambda$-bounded, then they satisfy
\[
\excond{Z_{e}^{t+1}-Z_{e}^{t}}{{\cal F}^{t}} \le \lambda\cdot\excond{S_{e}^{t+1}-S_{e}^{t}}{{\cal F}^{t}}.
\]
\end{lem}
\begin{proof}

Consider step $t+1$ of the process. We claim the following relationship
\[
\excond{S_{e}^{t+1}-S_{e}^{t}}{{\cal F}^{t}}=\frac{Y_{e}^{t}}{n-t}.
\]
We check this relation by a case-work. If $Y_{e,i}^{t}=0$, then we
already know the fate of $e$. In this case we either have $S_{e,i}^{t}=S_{e,i}^{t+1}=0$
if $e$ has been blocked, or we have $S_{e,i}^{t}=S_{e,i}^{t+1}=1$,
if we have taken $e$ before step $t$. In both cases left-hand
side and right-hand side are equal 0. Now if $Y_{e}^{t}=1$, then
we know that 1) we have not chosen $e$ in line~\ref{algline:abstact-cr-pick}
 before, and 2) $(e,C_{e})$ has not been blocked.
Then we can
 pick $e$ in step $t$ with probability $\frac{1}{n-t}$, what
means exactly that $S_{e}^{t}=0$ but $S_{e}^{t+1}=1$. 
The claim follows.
\end{proof}

\begin{lem}
\label{lem:martingale}
Suppose characteristic sequences of $e$ are $\lam$-bounded.
Then process $\br{(1+\lam)\cdot S_{e}^{t}+Y_{e}^{t}}_{t=0}^n$
is a submartingale.
\end{lem}
\begin{proof}
Recall that $Y_{e}^{t}=1-S_{e}^{t}-Z_{e}^{t}$.
From Lemma~\ref{lem:deltas} we have
\begin{eqnarray*}
\excond{Y_{e}^{t}-Y_{e}^{t+1}}{{\cal F}^{t}} & = & \excond{S_{e}^{t+1}+Z_{e}^{t+1}-S_{e}^{t}-Z_{e}^{t}}{{\cal F}^{t}} = \\
 & = & \excond{S_{e}^{t+1}-S_{e}^{t}+Z_{e}^{t+1}-Z_{e}^{t}}{{\cal F}^{t}} \le \\
 & \le & (1+\lam)\cdot\excond{S_{e}^{t+1}-S_{e}^{t}}{{\cal F}^{t}},
\end{eqnarray*}
\begin{eqnarray*}
\excond{\br{(1+\lam)\cdot S_{e}^{t+1}+Y_{e}^{t+1}}-\br{(1+\lam)\cdot S_{e}^{t}+Y_{e}^{t}}}{{\cal F}^{t}} =&\\
= \quad\excond{(1+\lam)\cdot\br{S_{e}^{t+1}-S_{e}^{t}}-\br{Y_{e}^{t}-Y_{e}^{t+1}}}{{\cal F}^{t}} \geq & 0,
\end{eqnarray*}
which means that the process $\br{(1+\lam)\cdot S_{e}^{t}+Y_{e}^{t}}_{t=0}^n$
is indeed a submartingale.
\end{proof}

\begin{lem}
\label{lem:cr}
Suppose a random-order CR scheme yields a controller mechanism with $\lam$-bounded characteristic sequences.
Then the probability that $e$ does not get blocked before it is picked is at least $\frac{1}{1 + \lam}$.
\end{lem}
\begin{proof}
Lemma~\ref{lem:martingale} guarantees that process $\br{(1+\lam)\cdot S_{e}^{t}+Y_{e}^{t}}_{t=0}^n$
is a submartingale.
Let $\tau=\min\setst t{Y_{e}^{t}=0}$ denote the
first moment when we get to know what happens with $e$. Since $\tau$
is a bounded (always $\tau\leq n$) stopping time, we can take advantage of the
Doob's stopping theorem to get
\[
\ex{(1+\lam)\cdot S_{e}^{0}+Y_{e}^{0}}\leq\ex{(1+\lam)\cdot S_{e}^{\tau}+Y_{e}^{\tau}}.
\]
Since $S_{e}^{0}=0=Y_{e}^{\tau}$ and $Y_{e}^{0}=1$, we have
\[
1=\ex{(1+\lam)\cdot S_{e}^{0}+Y_{e}^{0}}\leq\ex{(1+\lam)\cdot S_{e}^{\tau}+Y_{e}^{\tau}}=(1+\lam)\cdot\ex{S_{e}^{\tau}},
\]
and so $\ex{S_{e}^{\tau}}\geq\frac{1}{1+\lam}$.
Now one just has to note that $\prcond{e\mbox{ is available to take when picked}}{C_{e}}$
is exactly equal to $\ex{S_{e}^{\tau}}$
(conditioning on $C_{e}$ comes from the fact that the derivation
is performed this particular controller).
Since this holds for any choice of the controller $C_{e}$,
we get the same bound unconditionally.
\end{proof}

Thus the probability that element $e$ will be taken into the solution under condition $e \in R(x)$ is at least $\frac{1}{1+\lam}$.
By combining Corollary~\ref{lem:matroid-lam} and Lemma~\ref{lem:cr} we finish the proof of Theorem~\ref{thm:matroid-cr}.

\subsection{Combining constraints}
\label{sec:matroid-multi-cr}
Suppose now that we are given $k$ constraints ${\cal I}_{1},{\cal I}_{2},\dots, {\cal I}_{k}$.
The combination of the mechanisms is simple.
We assign each element $k$ controllers independently with respect to each constraint.
We scan elements in a random order and when an element gets accepted we independently update each controller mechanism.
An~element gets blocked if it is blocked in at least one constraint.

The correctness of the mechanism, i.e., the fact that we return a
set that is independent in all constraints, is clear. We need to argue
for the approximation ratio to be proper.
Let us refer to the characteristic
sequences of the $i$-th constraint
as $(^{i}S_{e}^{t}),\, (^{i}Z_{e}^{t}),\, (^{i}Y_{e}^{t})$.
In order to construct the characteristic
sequences describing the joint mechanism,
observe that an element gets blocked if at least one of its controllers gets blocked, it gets accepted if it is accepted in all constraints, and we get to know its fate if it is revealed in at least one constraint.
Recall that $Y_e^t = 1$ stands for fate of $e$ \textbf{not being} revealed before step $t+1$.
This can be summarized as
\begin{eqnarray*}
Z_{e}^{t}&=&\max\br{^{1}Z_{e}^{t},{}^{2}Z_{e}^{t},...,{}^{k}Z_{e}^{t}}, \\
S_{e}^{t}&=&\min\br{^{1}S_{e}^{t},{}^{2}S_{e}^{t},...,{}^{k}S_{e}^{t}}, \\ 
Y_{e}^{t}&=&\min\br{^{1}Y_{e}^{t},{}^{2}Y_{e}^{t},...,{}^{k}Y_{e}^{t}}.
\end{eqnarray*}
We call these the \emph{joint} characteristic sequences of $e$. The relationship between $Z_{e}^{t}$, $S_{e}^{t}$ and
$Y_{e}^{t}$ becomes again
$Y_{e}^{t}=1-S_{e}^{t}-Z_{e}^{t}$.

\begin{lem}
\label{lem:deltas-multi}
Suppose the characteristic sequences for the $i$-th constraint are $\lam_i$-bounded.
Then the joint characteristic sequences are $\br{\sum_i \lam_i}$-bounded.
\end{lem}
\begin{proof}
If $Y_e^t = 0$, then for some $i$ we have $^iY_e^t = 0$,
i.e., the fate of $e$ has been revealed in the the $i$-th constraint.
There are two cases: either  $^iZ_e^t = 1$ or  $^iS_e^t = 1$.
In the first case we have $Z_e^t = Z_e^{t+1} = 1$.
If $^iS_e^t = 1$, then the element $e$ has been picked before step $t$ and either it got accepted in all constraints or it had been blocked before in some other constraint.
In both cases be have $Z_e^t = Z_e^{t+1}$.

If $Y_e^t = 1$, then it holds $^iY_e^t = 1$ for all $i$.
We estimate the probability of any event $\br{Z_{e}^{t+1} > Z_{e}^{t}}$ by the union bound, obtaining
\begin{eqnarray*}
\excond{Z_{e}^{t+1}-Z_{e}^{t}}{{\cal F}^{t}} & = & \excond{\max\br{^{1}Z_{e}^{t+1},{}^{2}Z_{e}^{t+1},...,{}^{k}Z_{e}^{t+1}}-\max\br{^{1}Z_{e}^{t},{}^{2}Z_{e}^{t},...,{}^{k}Z_{e}^{t}}}{{\cal F}^{t}} \le \\
 & \leq & \excond{\max\br{^{1}Z_{e}^{t+1}-{}^{1}Z_{e}^{t},{}^{2}Z_{e}^{t+1}-{}^{2}Z_{e}^{t},...,{}^{k}Z_{e}^{t+1}-{}^{k}Z_{e}^{t}}}{{\cal F}^{t}} \le \\
 & \le & \excond{\sum_{i}{}^{i}Z_{e}^{t+1}-{}^{i}Z_{e}^{t}}{{\cal F}^{t}} = \sum_{i}\excond{{}^{i}Z_{e}^{t+1}-{}^{i}Z_{e}^{t}}{{\cal F}^{t}} \le  \\
 & \le & \sum_i \frac{\lam_i\cdot ^iY_e^t}{n-t} =  \frac{\sum_i\lam_i\cdot Y_e^t}{n-t}.
\end{eqnarray*}
\end{proof}

\crscheme*
\begin{proof}
The claims follows from Corollary~\ref{lem:matroid-lam} and Lemmas~\ref{lem:cr} and~\ref{lem:deltas-multi}.
\end{proof}

\section{Multi-parameter mechanism design}
\label{sec:bmumd}

Recall that each client $i \in\cal I$ is interested in purchasing one service from ${\cal J}_i$ and their valuation of an item $c \in {\cal J}_i$ is modeled by a random variable $v_c$, independent of other valuations, with a known distribution $D_c$.
Following~\cite{DBLP:conf/ipco/GuptaN13} we assume that the distribution $D_c$ is always discrete and takes values over ${\cal B} = \set{0,1,\ldots,B}$.

\subsection{Bounding by auction with copies}
Imagine
a setting where for each item $c \in {\cal J}_{i}$ we create an independent copy-client $c$ interested solely in this item. The new
instance has the same constraint system as the original one plus additional partition matroid.
We rely on the crucial lemma by Chawla et al.~\cite{DBLP:conf/stoc/ChawlaHMS10}, saying
that the optimal revenue in the new instance can be only greater 
because the competition increases.

This observation allows us to obtain an LP upperbound for the true OPT.
The linear program \textsc{Bmumd-LP}~\cite{DBLP:conf/ipco/GuptaN13} models the auction with copy-clients, which is single-parameter.
$C$ denotes the set of copy-clients, which is equivalent to the set of items, and ${\cal P}$ is the polytope of the constraint system.
\begin{eqnarray*}
\max & \sum_{c\in C}\sum_{p}x_{c,p}\cdot p\cdot\pr{v_{c}\geq p}&\quad\mbox{ \textsc{(Bmumd-LP)}}\\
\mbox{s.t.} & \br{\sum_{p}x_{c,p}\cdot\pr{v_{c}\geq p}}_{c}\in{\cal P}\\
 & \sum_{p}x_{c,p}\leq1 & \forall c\in C\\
 & \sum_{c\in{\cal J}_{i}}\sum_{p}x_{c,p}\cdot\pr{v_{c}\geq p}\leq1 & \forall i\in {\cal I}
\end{eqnarray*}

\begin{lem}[\cite{DBLP:conf/stoc/ChawlaHMS10, DBLP:conf/ipco/GuptaN13}]
\label{lem:chawla}
The optimal value of \textsc{Bmumd-LP} is an upper bound for
the maximal revenue in the multi-parameter auction.
\end{lem}

\subsection{Single client routine}

The algorithm scans clients in random order, and presents a price
menu to each client, from which the client picks one item which gives
him the highest utility, or resigns from choosing if all utilities are negative.
Such a procedure clearly yields a truthful mechanism.
Let $x_{c,p}$ be the probability that we place item $c \in {\cal J}_{i}$ with price $p$ in the menu of client $i$.
The vector $\mathbf{x} = (x_{c,p})$ describing randomized menu for client $i$ must satisfy following constraints. We will call it a \emph{menu-vector}.
\hspace{-1.5cm}
\begin{eqnarray*}
 & \sum_{p}x_{c,p}\leq1 & \forall c\in {\cal J}_{i}\\
 & \sum_{c\in{\cal J}_{i}}\sum_{p}x_{c,p}\cdot\pr{v_{c}\geq p}\leq1
\end{eqnarray*}

Given menu-vector $\mathbf{x}$, we construct the menu as follows.
Independently for each item $c$ we choose price $p$ with probability $x_{c,p}$ and discard the item with probability $1 - \sum_p x^t_{c,p}$.
Then the client reveals their utilities for each item.
We define $\evx {} c p$ to be the event of the item $c$ with price $p$ being at the top of the menu.
To ensure that it is well-defined we need to fix a mechanism to break the ties between items of equal utility to the client, e.g., lexicographically or by random choice. 
However we do not need to know the mechanism explicitly for the analysis sake.

The following lemma describes how to construct a menu-vector with almost tight guarantees on probabilities of item acceptance and rejection.
The proof, based on $O(1/\eps^2)$ rounds of a local search procedure, is located in Section~\ref{sec:single-client}.

\begin{restatable}{lem}{singleclient}
\label{lem:one-client-perfect}
Suppose we can compute values $\pr{\evx {} c p}$ in a polynomial time for a known menu-vector.
Then for any $\eps>0$ there is a polynomial-time procedure that, given menu-vector $\mathbf{x}$, finds another menu-vector $\mathbf{y}$, such that for each $c,p$:
\hspace{-1cm}
\[
\frac{1}{4}x_{c,p}\cdot\pr{v_{c}\geq p} \le \pr{\mathbf{Y}_{c,p}} \le \Big(\frac{1}{4} + \eps\Big)\cdot x_{c,p}\cdot\pr{v_{c}\geq p}.
\]
\end{restatable}

\subsection{The algorithm}

With the subroutine to handle a single client, we are ready to prove the main result of this paper.
\bmumd*

\begin{algorithm}
\caption{\label{alg:abstract-auction}Auction mechanism}
\begin{algorithmic}[1]
\STATE assign each item $c$ a controller $C_c$

\STATE \textbf{for} each client $i \in I$ in random order \textbf{do}

\STATE ~~~perform \textsc{SingleClientSubroutine} (Lemma~\ref{lem:one-client-perfect}) on the non-blocked
items in $J_i$

\STATE ~~~offer the chosen item $c$ to client $i$

\STATE ~~~\textbf{if} client $i$ accepts $c$ \textbf{then}  

\STATE ~~~\textbf{~~~}update controllers
\end{algorithmic} 
\end{algorithm}

We begin with a relaxation to \textsc{Bmumd-LP} and obtain vector $\br{x_{c,p}}_{c,p}$ that supplies the auction mechanism, that is based on the random-order contention resolution scheme.
The abstract view of the auction mechanism is presented in Algorithm~\ref{alg:abstract-auction}.

\paragraph{Matroid implementation} The controller mechanism for matroids
is analogous to the one from Theorem~\ref{thm:matroid-cr}.
We decompose vector $\br{\sum_{p}x_{c}^{p}\cdot\pr{v_{c}\geq p}}_{c\in\cal C}$ into a support in matroid $\M$, that is $\sum_{i\in\cal S}\beta_{i}\cdot B_{i}^{0}$ (Lemma~\ref{prelim:support}).
We find exchange-mappings $\phi\brq{B_{i}^{0},B_{j}^{0}}$
between each pair $i,j\in\cal S$ as in Lemma~\ref{prelim:mapping}.
Then for each element $c$ we choose $j(c)\in\cal S$ such that
$c\in B_{j}$, with probability $\frac{\beta_{j}}{\sum_{p}x_{c}^{p}\cdot\pr{v_{c}\geq p}}$,
call it the controller of $c$, and denote $C_{c}^t = B_{j(c)}^t$.

When an item $c \in J_i$ gets accepted by client $i$,
we update the controllers, as presented in Algorithm~\ref{alg:matroid-auction}.
A~pair $(c,\, C_c)$ gets blocked when $c$ is removed from $C_c$.

The correctness follows again from the invariant, that the set of served items is a subset of $B_{j}^{t}$, which is an independent set, for all $j \in\cal S$.

\begin{algorithm}
\caption{\label{alg:matroid-auction}Controller mechanism update for a matroid, restated}
\begin{algorithmic}[1]
\STATE \textbf{for} each $i\in{\cal S}:c\notin B_{i}^{t}$
\textbf{do}

\STATE ~~~\textbf{if }$\phi\brq{C_{c}^{t},B_{i}^{t}}\br c=\perp$
\textbf{then} $B_{i}^{t+1}\leftarrow B_{i}^{t}+c$

\STATE ~~~\textbf{if }$\phi\brq{C_{c}^{t},B_{i}^{t}}\br c=d$
\textbf{then} $B_{i}^{t+1}\leftarrow B_{i}^{t}-d+c$

\STATE \textbf{for} each $i\in{\cal S}:c\in B_{i}^{t}$
\textbf{do}

\STATE ~~~\textbf{$B_{i}^{t+1}\leftarrow B_{i}^{t}$}

\STATE find new exchange-mappings $\phi\brq{B_{i}^{t+1},B_{j}^{t+1}}$
between each pair $i,j\in\cal S$
\end{algorithmic} 
\end{algorithm}

\paragraph{Approximation guarantee}

We are interested in estimating the probability that a fixed item $c\in{\cal J}_i$
will be served to a client $i$ at price $p$ .
In this paragraph we condition all the events on the critical set $C_c$
and we argue that it will not get blocked until the turn of client $i$ with high probability.
We retrace the reasoning from Section~\ref{sec:character}
and assign each pair $(c, C_c)$ the characteristic sequences $(S_c^t, Z_c^t, Y_c^t)$.
This time $S_c^t = 1$ carries semantics of $c$ having ended up in the menu of client $i$ before step $t+1$.

\begin{lem}
\label{lem:matroid-auction-pr}
The characteristic sequences of the auction mechanism for a matroid are $\br{\frac{1}{4} + \eps}$-bounded.
\end{lem}
\begin{proof}
We proceed as in Lemma~\ref{lem:matroid-blocking-pr}.
We need to show that probability that $(c, C_c)$ gets blocked in turn $t$ in a single matroid is at most
\[
\br{\frac{1}{4} + \eps}\cdot\frac{1}{n-t}.
\]
A blocking event happens when an item $d$ of a different agent $j$
gets chosen, that makes $c$ removed from its controller set.
The agent $j$ is chosen with probability $\frac{1}{n-t}$.
Then the agent has to pick item $d$ from the menu.
The properties of the single-agent subroutine (Lemma~\ref{lem:one-client-perfect})
guarantees that this
happens with probability at most $\br{\frac{1}{4} + \eps}\sum_{p}x_{d}^{p}\cdot\pr{v_{d}\geq p}$.
Then $d$ must be assigned a controller $j\in\cal S$ which makes $c$ removed from $C_{c}$ -- this occurs with
probability $\frac{\beta_j}{\sum_{p}x_{d}^{p}\cdot\pr{v_{d}\geq p}}$.
The total probability of any of these events is
\begin{align*}
 & \frac{1}{n-t}\sum_{d}\br{\frac{1}{4} + \eps}\sum_{p}x_{d}^{p}\cdot\pr{v_{d}\geq p}\cdot\sum_{{j}:\phi_{[B_{j}^t,C_{c}^t]}\br d=c}\frac{\beta_{j}}{\sum_{p}x_{d}^{p}\cdot\pr{v_{d}\geq p}} =\\
= & \br{\frac{1}{4} + \eps}\frac{1}{n-t}\sum_{d}\sum_{{j}:\phi_{[B_{j}^t,C_{c}^t]}\br d=c}\beta_{j} \le \\
\leq &\br{\frac{1}{4} + \eps}\frac{1}{n-t},
\end{align*}
where the last inequality follows from Lemma~\ref{lem:matroid:support}.
\end{proof}

\begin{proof}[Proof of Theorem \ref{thm:bmumd-main}]
The joint mechanism for the intersection of $k$ matroids
is given by assigning each item $k$ controllers and blocking the item if at least one of its controllers has been blocked.
After an item is picked, it gets placed in the
menu as long as it has not been blocked before.
This leads to a construction analogous to the one
from Section~\ref{sec:matroid-multi-cr}.

In Lemma~\ref{lem:matroid-auction-pr}
we have analyzed a blocking event in a single turn and for a single matroid.
We use Lemma~\ref{lem:deltas-multi} to extend this result to $k$ matroids.
Then we apply Lemma~\ref{lem:cr} to get the global probability of not being blocked in any turn and in any matroid.
\[
\prcond{c\mbox{ does not get blocked until the turn of client }i}{C_{c}}\geq\frac{1}{k\cdot\br{\frac{1}{4} + \eps} + 1}
\]

Lemma~\ref{lem:one-client-perfect} guarantees that once $c$ gets into the menu of client $i$, it will be served at price $p$ with probability at least $\frac{1}{4}\cdot x_{c,p}\cdot \pr{v_{c}\geq p}$.
By multiplying these quantities and setting the final value of $\eps$ to $\frac{\eps}{4k}$ we obtain
\[
\prcond{c\mbox{ gets served client $i$ at price $p$}}{C_{c}}\geq\frac{\frac{1}{4}\cdot x_{c,p} \cdot \pr{v_{c}\geq p}}{k\cdot\br{\frac{1}{4} +\frac{\eps}{4k}} +1} = \frac{ x_{c,p} \cdot \pr{v_{c}\geq p}}{k +4 + \eps}.
\]

Since this holds for every choice of $C_c$ we get the same bound unconditionally.
It means that the expected revenue of the mechanism is at least
$\frac{1}{k+4+\eps}$ times the optimal value of the linear program \textsc{Bmumd-LP}, what finishes the proof.
\end{proof}

\section{Submodular optimization}
\label{sec:submodular}
As another elegant application of our toolbox, as well as a building block for submodular stochastic probing, we show a framework for non-negative submodular function maximization.
We are given the following optimization task over (possibly) a sequence of constraints $({\cal I}_i)_{i=1}^k$.
\begin{eqnarray*}
\max & f\br X\\
\mbox{s.t.} & X\in{\cal I}_{1}\cap{\cal I}_{2}\cap\dots \cap{\cal I}_{k}
\end{eqnarray*}
We first need to solve the multilinear relaxation for this
problem.
\begin{eqnarray*}
\max & F\br x\\
\mbox{s.t.} & x\in\P{{\cal I}_{1}}\cap\P{{\cal I}_{2}}\cap\dots\cap\P{{\cal I}_{k}}
\end{eqnarray*}

We can approximately solve such a optimization problem with the measured continuous greedy algorithm~\cite{DBLP:conf/focs/FeldmanNS11},
that provides a~vector $x$ such that $F\br x\geq(\frac{1}{e}-\eps)\cdot f\br{X_{OPT}}$.
We are going to sample elements with respect to $x$ and execute the random-order contention resolution scheme on the sampled set with a following postprocessing: we add an accepted element $e$ to $X$ only if $f(X \cup \{e\}) > f(X)$.
We claim that this procedure returns a solution $X$ such that
$\ex{f\br X} \geq\frac{1}{\lam+1}F\br x$.

\begin{lem}
\label{lem:submodular-abstract}
Consider a sampling scheme with $\lambda$-bounded characteristic sequences, in which the chosen element $e$ materializes with probability $x_e$.
Suppose we are given a non-negative submodular function $f$
and we accept the chosen element if it materializes and taking it increases the value of $f$.
Such a procedure generates a random set $X$ such that $\ex{f(X)} \ge \frac{1}{\lam +1}F(x)$. 
\end{lem}
\begin{proof}
This time we need to track globally the solution that we create, and not
just a particular element. Let $X^{t}$ be the solution created up to step
$t$ and $S^{t}=\setst{S_{e}^{t} = 1}{e \in E} \,\cap\, R(x)$.
For all $t$ it holds $X^t \subseteq S^t$.
Also let $Z^{t}$
be the set of all present elements that have been blocked up to step $t$,
that is, $Z^{t}=\setst{Z_{e}^{t} = 1}{e \in E} \,\cap\, R(x)$.
We are going to show that the following
sequence 
\[
\br{\br{\lam+1}\cdot f\br{X^{t}}-f\br{S^{t}\cup Z^{t}}}_{t=0}^n
\]
is a submartingale. 
Let us consider the deltas
\begin{eqnarray*}
\excond{f\br{X^{t+1}}-f\br{X^{t}}}{{\cal F}^{t}} & = & \sum_{e\in E}x_e\cdot\excond{S_{e}^{t+1}-S_{e}^{t}}{{\cal F}^{t}}\cdot\max\br{f\br{X^{t}+e}-f\br{X^{t}}, 0}, \\
\excond{f\br{S^{t+1}\cup Z^{t+1}}-f\br{S^{t}\cup Z^{t}}}{{\cal F}^{t}} & = & \sum_{e\in E}x_e\cdot\excond{S_{e}^{t+1}-S_{e}^{t}+Z_{e}^{t+1}-Z_{e}^{t}}{{\cal F}^{t}}\cdot\br{f\br{S^{t}\cup Z^{t}+e}-f\br{S^{t}\cup Z^{t}}}.
\end{eqnarray*}
From Lemma~\ref{lem:deltas} we have
$\excond{Z_{e}^{t+1}-Z_{e}^{t}}{{\cal F}^{t}}\leq \lam\cdot\excond{S_{e}^{t+1}-S_{e}^{t}}{{\cal F}^{t}}$,
and from submodularity we know that 
$f\br{S^{t}\cup Z^{t}+e}-f\br{S^{t}\cup Z^{t}}\leq f\br{S^{t}+e}-f\br{S^{t}} \le f\br{X^{t}+e}-f\br{X^{t}}.$
Therefore
\begin{eqnarray*}
&& \excond{f\br{S^{t+1}\cup Z^{t+1}}-f\br{S^{t}\cup Z^{t}}}{{\cal F}^{t}} = \\
& = & \sum_{e\in E}x_e\cdot\excond{S_{e}^{t+1}-S_{e}^{t}+Z_{e}^{t+1}-Z_{e}^{t}}{{\cal F}^{t}}\cdot\br{f\br{S^{t}\cup Z^{t}+e}-f\br{S^{t}\cup Z^{t}}} \le \\
& \le & \sum_{e\in E}x_e\cdot\excond{S_{e}^{t+1}-S_{e}^{t}+Z_{e}^{t+1}-Z_{e}^{t}}{{\cal F}^{t}}\cdot \max\br{f\br{X^{t}+e}-f\br{X^{t}}, 0} \le \\
 & \leq & \br{\lam+1}\cdot\sum_{e\in E}x_e\cdot\excond{S_{e}^{t+1}-S_{e}^{t}}{{\cal F}^{t}}\cdot\max\br{f\br{X^{t}+e}-f\br{X^{t}}, 0} =\\
 & = & \br{\lam+1}\cdot\excond{f\br{X^{t+1}}-f\br{X^{t}}}{{\cal F}^{t}},
\end{eqnarray*}
and we conclude that the sequence 
$\br{\br{\lam+1}\cdot f\br{X^{t}}-f\br{S^{t}\cup Z^{t}}}_{t=0}^n$
is indeed a submartingale. Since $S^n \cup Z^n = R(x)$ and $f(\emptyset) \ge 0$, we have
\begin{eqnarray*}
0 & \le & \lam\cdot f(\emptyset) = \br{\lam+1}\cdot f\br{X^{0}}-f\br{S^{0}\cup Z^{0}} \le \\
& \le & \br{\lam+1}\cdot\ex{f\br{X^{n}}}-\ex{f\br{S^{n}\cup Z^{n}}}=\br{\lam+1}\cdot\ex{f\br{X^{n}}}-\ex{f\br{R\br x}},
\end{eqnarray*}
and further, by the definition of the multilinear extension $F,$
\[
\br{\lam+1}\cdot\ex{f\br{X^{n}}}\geq\ex{f\br{R\br x}}=F\br x.
\]
\end{proof}
\submodular*
\begin{proof}
Corollary~\ref{lem:matroid-lam} and Lemma~\ref{lem:deltas-multi} implies that the characteristic sequences of the CR scheme for the intersection of $k$ matroids are $k$-bounded.
We 
find vector $x$, such that $F\br x\geq(\frac{1}{e}-{\eps})\cdot f\br{X_{OPT}}$,
with the measured continuous greedy algorithm~\cite{DBLP:conf/focs/FeldmanNS11},
sample elements accordingly to $x$, and apply Lemma~\ref{lem:submodular-abstract}.
In the end we adjust $\eps$ to appropriately depend on $k$.
\end{proof}

\section{Stochastic $k$-set packing}
\label{sec:packing}

In the basic stochastic $k$-set packing problem, we are given $n$ elements/columns, where each item $e\in E=\brq n$
has a profit $v_{e}\in\mathbb{R}_{+}$, and a random $d$-dimensional
size $L_{e}\in\{0,1\}^{d}$. The sizes are independent for different
items. Additionally, for each item $e$, there is a set $Q_{e}$ of
at most $k$ coordinates such that each size vector $L_{e}$ takes
positive values only in these coordinates, i.e., $L_{e}\subseteq Q_{e}$
with probability $1$. We are also given a~capacity vector $b\in\mathbb{Z}_{+}^{d}$
into which items must be packed. We assume that $v_{e}$ is a random
variable that may be correlated with $L_{e}$. The coordinates of
$L_{e}$ also might be correlated between each other.
After probing element $e$, its size $L_{e}$ is revealed
and the reward $v_{e}$ is drawn.

Equivalently, one can consider $d$ copies of each item: $e^1, e^2, \dots, e^d$, so that if $e$ is probed then its $i$-th copy materializes with probability $p_e^i$ and $p_e^i = 0$ for $i \not\in Q_e$.
In this view the capacity vector $b$ induces a constraint family $U_{b_i}$ over the ground set $E^i$ of $i$-th copies of each item.
We can easily generalize this setting to consider arbitrary matroid constraint $\M_i$ over $E^i$.
Let $R^i \subseteq E^i$ denote the random set of  materialized $i$-th copies of elements.

The following linear program (used first in~\cite{Bansal:woes} in the case of uniform matroids) provides a relaxation for
the problem.\micr{referencja albo argument ze to relaksacja}
\begin{eqnarray*}
\max\qquad\sum_{e\in E}\ex{v_{e}}\cdot x_{e} &  & \qquad\mbox{\textsc{(SetPacking-LP)}}\\
\mbox{s.t.}\qquad p^{i}\cdot x\in\P{\M_i} &  & \qquad\forall i\in\brq d\\
x_{e}\in\brq{0,1} &  & \qquad\forall e\in\brq n,
\end{eqnarray*}
where, as usual, $x_{e}$ is interpreted as $\pr{\mbox{optimal solution probes  }e}$.
We are going to present a probing strategy in which for every element
$e$ the probability of being probed is at least $\frac{x_{e}}{k+1}$.
Having this property, a $(k+1)$-approximation guarantee will follow.

\packing*
\begin{proof}
We  present the abstract view of the mechanism in Algorithm~\ref{alg:Stochastic-contention-resolution-k-set-packing}.
\begin{algorithm}
\caption{Controller mechanism for stochastic $k$-set packing \label{alg:Stochastic-contention-resolution-k-set-packing}}

\begin{algorithmic}[1]

\STATE  solve \textsc{SetPacking-LP}; let $\br{x_{e}}$ be the solution

\STATE  \textbf{for} each element $e$  \textbf{do}
\STATE  ~~~\textbf{for} each constraint $i\in Q_e$ \textbf{do}
\STATE  ~~~~~~assign $e$ a controller $C_e^i$ with respect to vector $p^{i}\cdot x$

\STATE  \textbf{for} each element $e$ in random order \textbf{do}\\

\STATE  ~~~\textbf{if} $(e^i, C_e^i)$ has been blocked for any $i \in Q_{e}$ \textbf{then}
\STATE  ~~~~~~continue

\STATE  ~~~take $e$ into the solution with probability $x_e$ \label{algline:probing-accept}

\STATE  ~~~\textbf{for} each constraint $i \in Q_{e}$ \textbf{do }

\STATE  ~~~~~~\textbf{if} $e^i \in R^i$ \textbf{then}

\STATE  ~~~~~~~~~update controllers with respect to the $i$-th constraint

\end{algorithmic}
\end{algorithm}

\paragraph{Matroid implementation} The controller mechanism
is analogous to those in Sections~\ref{sec:matroid-cr} and~\ref{sec:bmumd}.
We decompose vector $p^{i}\cdot x$ into a support in matroid $\M_i$,
 that is $p^{i}\cdot x = \sum_{j\in{\cal S}_i}\beta_{j}\cdot B_{j}^{0}$ (Lemma~\ref{prelim:support}),
and find exchange-mappings
between each pair of sets (Lemma~\ref{prelim:mapping}).
Then for each element $e$ we choose $j(e,i) \in{\cal S}_i$ such that
$e \in B_{j}$, with probability $\frac{\beta_{j}}{p_e^i\cdot x_e}$,
call it the $i$-th controller of $c$, and denote\footnote{
There is a notation conflict in the superscript of $C_{e}^i$ as in previous sections we used that to refer to the set $C_e$ in step $t$. This time we reserve it to denote the constraint index.} it by~$C_{e}^i$.

When an element $e$ gets accepted,
we update the controllers in $Q_e$, as presented previously in Algorithm~\ref{alg:matroid-auction}.
A~pair $(e,\, C_e^i)$ gets blocked when $e$ is removed from $C_e^i$.

\paragraph{Correctness}
Let $S^t$ denote the set of accepted elements up to step $t$.
The controller mechanism ensures that if $e$ has not been blocked, then $\setst{e^i}{e\in S^t} \,\cap\, R^i$ together with $e$ belong to $C_e^i$, which is an independent set in $\M_i$ for all $i \in Q_e$.
Therefore taking $e$ into the solution would not break any constraint from $Q_e$ and other constraints are oblivious to $e$.

\paragraph{Approximation guarantee}
Consider the event that $(e, C_e^i)$ gets blocked in step $t$ in the $i$-th constraint.
For this to happen, an element $f\ne e$ must be chosen with probability $\frac{1}{n-t}$, it must be taken into the solution with probability $x_f$ and it must exist in $R^i$ with probability $p_f^i$.
A choice of particular controller $j(f,i)$ happens with probability $\frac{\beta_{j}}{p_e^i\cdot x_e}$.
Let $\Gamma^{i,t}\br{e,C}$ denote the set of pairs $(f,j)$ that would block $(e,C)$ in step $t$ in $i$-th constraint, as in Lemma~\ref{lem:matroid:support}.
Combining all of these, we get a bound on the probability of a blocking event
\begin{eqnarray*}
&& \frac{1}{n-t}\sum_{f\ne e}\pr{f\mbox{ gets accepted in line \ref{algline:probing-accept}}}\cdot\pr{f\in R^i}\sum_{j:\br{f,j}\in\Gamma^{i,t}\br{e,C^i_e}}\pr{f\mbox{ chooses controller }j} =\\
&=& \frac{1}{n-t}\sum_{f\ne e} x_f\cdot p_f^i \sum_{j:\br{f,j}\in\Gamma^{i,t}\br{e,C^i_e}}  \frac{\beta_{j}}{p_f^i\cdot x_f} = \\
&=& \frac{1}{n-t}\sum_{\br{f,j}\in\Gamma^{i,t}\br{e,C^i_e}}\beta_{j} \le \frac{1}{n-t},
\end{eqnarray*}
where the last inequality follows from Lemma~\ref{lem:matroid:support}.

As in previous sections, we use characteristic sequences 
to keep track of status of $e$ in the $i$-th constraint.
From the derivation above we know that they are 1-bounded.
Since $e$ could be blocked only by constraints from $Q_e$,
the joint characteristic sequences are given by a combination of at most $k$ sequences, as described in Section~\ref{sec:matroid-multi-cr}.
Lemma~\ref{lem:deltas-multi} says that these sequences are $k$-bounded and Lemma~\ref{lem:cr} guarantees that $e$ will reach line~\ref{algline:probing-accept} with probability at least $\frac{1}{k+1}$.
This finishes the proof.
\end{proof}
\section{Submodular stochastic probing}
\label{sec:probing}
Recall that each element $e$ of the universe $E$ might be active 
with probability $p_{e}$ and the only way to learn whether $e$ is active or not is to probe it.
If an element is active, the algorithm is forced to add it to the
current solution. In this way, the algorithm gradually constructs
a solution consisting of active elements.

We are given two independence systems of downward-closed
sets: an \emph{outer} independence system $\br{E,\Iout}$ restricting
the set of elements probed by the algorithm, and an \emph{inner} independence
system $\br{E,\Iin}$, restricting the set of elements taken by the
algorithm. We denote by $Q^{t}$ the set of elements probed in the
first $t$ steps of the algorithm, and by $S^{t}$ the subset of active
elements from $Q^{t}$. Then, $S^{t}$ is the partial solution constructed
by the first $t$ steps of the algorithm. We require that at each
time $t$, $Q^{t}\in\Iout$ and $S^{t}\in\Iin$. Thus, at each time
$t$, the element $e$ that we probe must satisfy both $Q^{t-1}\cup\{e\}\in\Iout$
and $S^{t-1}\cup\{e\}\in\Iin$. The goal is to maximize expected value
$\ex{f\br {S^n}}$ where $f$ is a given non-negative submodular function.
We denote such
a stochastic probing problem by $\br{E,p,\Iin,\Iout}$ with function
$f$ stated on the side (if needed).

In the first presentation of the mechanism we neglect function $f$ and focus on ensuring that each element will be probed with sufficiently high probability, i.e., we maximize linear objectives.
Later we introduce a~relaxation for the non-negative submodular case and combine the algorithm with an argument from Section~\ref{sec:submodular}.

\subsection{Sampling scheme}
\begin{algorithm}
\caption{Controller mechanism for submodular stochastic probing \label{alg:probing-easy}}

\begin{algorithmic}[1]

\STATE  solve \textsc{Probing-MP}; let $x = \br{x_{e}}_{e\in E}$ be the solution and $R(x)$ be a set sampled with respect to $x$
\STATE  \textbf{for} each element $e$  \textbf{do}
\STATE  ~~~\textbf{for} each constraint $i\in K_{out}$ \textbf{do}
\STATE  ~~~~~~assign $e$ a controller $C_e^i$ with respect to vector $x$
\STATE  ~~~\textbf{for} each constraint $i\in K_{in}$ \textbf{do}
\STATE  ~~~~~~assign $e$ a controller $C_e^i$ with respect to vector $p\cdot x$

\STATE  \textbf{for} each element $e$ in random order \textbf{do}

\STATE  ~~~\textbf{if} $e\not\in R(x)$ \textbf{ then}\label{algline:delta}
\STATE  ~~~~~~continue

\STATE  ~~~\textbf{if} $(e, C_e^i)$ has been blocked for any constraint $i$ \textbf{then}
\STATE  ~~~~~~continue

\STATE  ~~~\textbf{for} each constraint $i \in K_{out}$ \textbf{do }

\STATE  ~~~~~~update controllers with respect to the $i$-th constraint

\STATE  ~~~probe $e$ with probability of success $p_e$ 

\STATE  ~~~\textbf{if} probe has been successful \textbf{then}

\STATE  ~~~~~~\textbf{for} each constraint $i \in K_{in}$ \textbf{do }

\STATE  ~~~~~~~~~update controllers with respect to the $i$-th
constraint
\end{algorithmic}
\end{algorithm}

\begin{lem}
\label{lem:probing-easy}
Given a vector $x = (x_e)_{e\in E}$, we can construct a stochastic probing  procedure with $(k_{in} + k_{out})$-bounded characteristic sequences,
in which each chosen element is taken into solution with probability $p_e\cdot x_e$.
\end{lem}
\begin{proof}
The mechanism is described in Algorithm~\ref{alg:probing-easy},
where $K_{in}$ and $K_{out}$ represent families of respectively inner and outer constraints.

\paragraph{Matroid implementation}
The implementation is again analogous to the one from Section~\ref{sec:matroid-cr}, however there are minor details in handling inner and outer constraints.
First, observe that for inner constraints we use decomposition  
$p\cdot x = \sum_{j\in{\cal S}_i}\beta_{j}\cdot B_{j}^{0}$,
and for outer constraints it is, as usual,
$x = \sum_{j\in{\cal S}_i}\beta_{j}\cdot B_{j}^{0}$.
Then the choice of a controller is performed according to respectively 
$\frac{\beta_{j}}{p_e\cdot x_e}$ and $\frac{\beta_{j}}{x_e}$.

When an element $e$ is selected according to the random permutation, we first check if it belongs to the sampled set $R(x)$.
Since we are going to probe it, we update controllers for the outer constraints, possibly blocking other elements from being probed.
Then, if the probe turned out successful, we update controllers
for the inner constraints.

\paragraph{Correctness}
The controller mechanism ensures that $Q^t$ is a subset of all controller sets, which are independent, in the outer constraints and likewise for $S^t$.
As long as element $e$ belongs to $C_e^i$ for all $i \in K_{in} \cup K_{out}$, adding $e$ to both $S^t$ and $Q^t$ would not break any constraint.

\paragraph{Approximation guarantee}
We claim that the probability of a blocking event for $e$ in step $t$ in any single constraint is at most $\frac{1}{n-t}$.
For outer constraints, the derivation is analogous as in Section~\ref{sec:matroid-cr}.
For inner constraints we bound the probability by
\begin{eqnarray*}
&& \frac{1}{n-t}\sum_{f\ne e}\pr{f\in R(x)}\cdot\pr{f\mbox{ gets probed successfully}}\sum_{j:\br{f,j}\in\Gamma^{i,t}\br{e,C^i_e}}\pr{f\mbox{ chooses controller }j} =\\
&=& \frac{1}{n-t}\sum_{f\ne e} x_f\cdot p_f\sum_{j:\br{f,j}\in\Gamma^{i,t}\br{e,C^i_e}}  \frac{\beta_{j}}{p_f\cdot x_f} = \\
&=& \frac{1}{n-t}\sum_{\br{f,j}\in\Gamma^{i,t}\br{e,C^i_e}}\beta_{j} \le \frac{1}{n-t},
\end{eqnarray*}
where the last inequality follows from Lemma~\ref{lem:matroid:support}.
We conclude that for each constraint the characteristic sequences of $e$ are 1-bounded.
Lemma~\ref{lem:deltas-multi} guarantees that the joint characteristic sequences are $(k_{in} + k_{out})$-bounded.
\end{proof}

\subsection{Relaxation for a non-negative submodular objective}
So far we were using only the multilinear relaxation $F$ of a submodular function $f$. It was mainly due to the convenient fact that $F(x)$ is exactly equal to $\ex{f(\R)}$,
i.e., it corresponds to sampling each point $e\in E$ independently
with probability $x_{e}$.
Here it will also be used to guide the algorithm, however we shall need another relaxation of a submodular function to get an appropriate benchmark.

Another extension of $f$ studied in \cite{Calinescu2011} is given
by: 
\[
f^{+}(y)=\max\setst{\sum_{A\subseteq E}\alpha_{A}f(A)}{\sum_{A\subseteq E}\alpha_{A}\le1,\ \forall_{j\in E}\sum_{A:j\in A}\alpha_{A}\le y_{j},\ \forall_{A\subseteq E}\,\alpha_{A}\ge0}.
\]
Intuitively, the solution $\br{\alpha_{A}}{}_{A\subseteq E}$ above
represents the distribution over $2^{E}$ that maximizes the value
$\ex{f(A)}$ subject to a constraint that its marginal values satisfy
$\pr{i\in A}\le y_{i}$. The value $f^{+}\br y$ is then the expected
value of $\ex{f\br A}$ under this distribution, while the value of
$F\br y$ is the value of $\ex{f\br A}$ under the particular distribution
that places each element $i$ in $A$ independently. This relaxation
is important for our applications because the following mathematical
programming relaxation gives an upper bound on the expected value
of the optimal feasible strategy for the submodular stochastic probing
problem: 
\begin{align*}
\textrm{maximize } &\setst{f^{+}\br{x\cdot p}}{x\in \P{\Iin,\Iout}},\quad\quad\mbox{(\textsc{Probing-MP})}  \\
\mbox{where } &
\P{\Iin,\Iout}=\setst x{ x\in \P{\Iout},p\cdot x\in\P{\Iin},x\in\brq{0,1}^{E}}. \nonumber 
\end{align*}

\begin{lem}
\label{lem:math-programming-bound} Let $S$ be the (random) solution generated by the optimal
strategy for the stochastic probing problem with non-negative submodular objective function $f$ over $\br{E,p,\Iin,\Iout}$.
Then $\ex{f\br{S}}\le f^{+}\br{x^{+}\cdot p}$, where $x^{+}=\mbox{argmax}_{y\in\P{\Iin,\Iout}}f^{+}\br{y\cdot p}$. \end{lem}
\begin{proof}
Denote the optimal probing strategy by $\cal S$.
We construct a feasible solution $x$ to \textsc{Probing-MP}
by setting $x_{e}=\pr{{\cal S} \mbox{ probes }e}$. First, we show that
this is indeed a feasible solution.
The set of elements $Q$ probed by any execution of $\cal S$ is always
an independent set of each outer matroid ${\cal M}_{j}^{out}=\br{E,\Iout_{j}}$,
i.e.\ $Q\in\bigcap_{j=1}^{\kout}\Iout_{j}$. Thus the vector $\ex{\chr Q}=x$
may be represented as a convex combination of vectors from $\setst{\chr A}{A\in\bigcap_{j=1}^{\kout}\Iout_{j}}$,
and so $x\in{\cal P}\br{\M_{j}^{out}}$ for any $j\in K_{out}$.
Analogously, the set of elements $S$ that were successfully probed
by $\cal S$ satisfy $S\in\bigcap_{j=1}^{\kin}\Iin_{j}$ for every possible
execution of $\cal S$. Hence, the vector $\ex{\chr S}=x\cdot p$ may
be represented as a convex combination of vectors from $\setst{\chr A}{A\in\bigcap_{j=1}^{\kin}\Iin_{j}}$
and so $x\cdot p\in{\cal P}\br{{\M_{j}^{in}}}$ for any $j\in K_{in}$.
The value $f^{+}(x\cdot p)$ gives the maximum value of $\exls{S\sim\mathcal{D}}{f(S)}$
over all distributions $\mathcal{D}$ satisfying $\prls{S\sim\mathcal{D}}{e\in S}\leq x_{e}\cdot p_{e}$.
The solution $S$ returned by $\cal S$ satisfies $\pr{e\in S}=\pr{{\cal S}\textrm{ probes \ensuremath{e}}}\cdot p_{e}=x_{e}\cdot p_{e}$.
Thus, $\cal S$ defines one such distribution, and so we have $\ex{f(S)}\le f^{+}(x\cdot p)\le f^{+}(x^{+}\cdot p)$. 
\end{proof}

We have obtained a relaxation, but it relies on  $f^+$, evaluation of which is already NP-hard. Optimization over~$F$ alone is not enough, since from the above discussion we know that for any point $x$ we have $f^{+}(x)\geq F(x)$.
Hence, we need another tool to use $F$ for optimization, but with provable guarantees over the solution of \textsc{Probing-MP}.
The following lemma states a stronger lower bound for
the measured greedy algorithm of Feldman et al.~\cite{DBLP:conf/focs/FeldmanNS11}.
The proof is postponed to Section~\ref{sec:new-bound-measured}.
\begin{lem}
\label{lem:measured-greedy} Let $b\in\brq{0,1}$, let $f$ be a non-negative submodular
function with multilinear extension $F$, and let $\mathcal{P}$ be
any downward closed polytope. Then, the solution $x\in\brq{0,1}^{E}$
produced by the measured greedy algorithm satisfies 1) $x\in b\cdot{\cal P}$,
2) $F\br x\ge\br{b\cdot e^{-b}-\eps}\cdot\max_{x\in\mathcal{P}}f^{+}\br x$.
\end{lem}

\begin{cor}
\label{cor:measured-greedy}
We can find a vector $x = (x_e)_{e\in E} \in \P{\Iin,\Iout}$,
such that $F(p\cdot x)$ is no less than  $\frac{1}{e} - \eps$ times the optimum of \textsc{Probing-MP}.
The procedure runs in polynomial time for any $\eps > 0$.
\end{cor}
\begin{proof}
First observe we can neglect elements with $p_e = 0$.
We substitute $y= x \cdot p$ and the polytope
$\P{\Iin,\Iout}$ becomes $\setst y{ y / p \in \P{\Iout},y\in\P{\Iin},y/p \in\brq{0,1}^{E}}$.
We optimize $f^+(y)$ over such polytope using Lemma~\ref{lem:measured-greedy} with $b=1$.
\end{proof}

\probing*
\begin{proof}
First we find a vector $x = (x_e)_{e\in E} \in \P{\Iin,\Iout}$, so that $F(p\cdot x)$ is
no less than $(\frac{1}{e} - \eps)$ times the optimum of \textsc{Probing-MP}, which is no less than the optimal revenue of the optimal probing mechanism (Lemma~\ref{lem:math-programming-bound} and Corollary~\ref{cor:measured-greedy}).
We run Algorithm~\ref{alg:probing-easy} on vector $x$ with a minor modification: in line~\ref{algline:delta} we check whether $f(S^t\cup\{e\}) > f(S^t)$, since the function $f$ does not have to be monotone.
Lemma~\ref{lem:probing-easy} combined with Lemma~\ref{lem:submodular-abstract} guarantee that the (random) set $S^n$ of successfully probed elements satisfies $\ex{f(S^n)} \ge \frac{1}{k_{in}+k_{out}+1} F(p\cdot x)$.
In the end we adjust $\eps$ to appropriately depend on $k_{in}+k_{out}$.
\end{proof}

We remark that the same machinery works for  stochastic $k$-set packing and we also can replace its linear objective wit a non-negative submodular function.

\subsection{Stronger bound for the measured continuous greedy algorithm}
\label{sec:new-bound-measured}
The results of this section are due to Justin Ward~\cite{Justin15private}.

We now briefly review the measured
continuous greedy algorithm of Feldman et al. \cite{DBLP:conf/focs/FeldmanNS11}.
The algorithm runs in $1 / \delta$ discrete time steps within time interval $[0,T]$, where $T \le 1$ and $\delta$
is a suitably chosen parameter depending on $n = |E|$.
Denote
\[
\partial_{e}F(x) = \frac{F\br{x\lor\chr e}-F(x)}{1-x_e},
\]
where $\lor$ stands for element-wise maximum. 

 Let $y\br t$ be the
current fractional solution at time $t$.
In each step the algorithm
selects vector $I\br t\in{\cal P}$ given by $\arg\max_{x\in{\cal P}}\sum_{e\in E}x_{e}\cdot\br{F\br{y\br t\lor\chr{e}}-F\br{y\br t}}$.
Then, it sets $y_{e}\br{t+\delta}=y_{e}\br t+\delta\cdot I_{e}(t)\cdot(1-y_{e}(t))$
and moves on to time $t+\delta$.

The analysis of Feldman et al.\ shows that if, at every time step
\begin{equation}
F(y(t+\delta))-F(y(t))\ge\delta\cdot\brq{e^{-t}\cdot f\br{OPT}-F\br{y(t)}}-O\br{n^{3}\delta^{2}f\br{OPT}},\label{eq:feldman-main}
\end{equation}
then for appropriate choice of $\delta$ we have $F\br{y(T)}\ge\brq{Te^{-T}-\eps}\cdot f(OPT).$
We note that, in fact, this portion of their analysis works even if
$f(OPT)$ is replaced by any constant value. Thus, in order to prove
our claim, it suffices to derive an analogue of \eqref{eq:feldman-main}
in which $f(OPT)$ is replaced by $f^{+}(x^{+})$, where $x^{+}=\mbox{argmax}_{y\in\mathcal{P}}f^{+}\br y$.
The remainder of the proof then follows as in~\cite{DBLP:conf/focs/FeldmanNS11}.

Lemma~\ref{lem:feldman-generalization} below contains the required
analogue of \eqref{eq:feldman-main}. Hence it implies Lemma~\ref{lem:measured-greedy}.
\begin{lem}
\label{lem:feldman-generalization} For every time $0\le t\le T$

$F\br{y(t+\delta)}-F\br{y(t)}\ge\delta\cdot\brq{e^{-t}\cdot f^{+}(x^{+})-F(y(t))}-O(n^{3}\delta^{2})\cdot f^{+}\br{x^{+}}.$
\end{lem}
We shall require the
following additional facts from the analysis of \cite{DBLP:conf/focs/FeldmanNS11}.
\begin{lem}[Lemma 3.3 in \cite{DBLP:conf/focs/FeldmanNS11}]\label{lem:feldman-1} Consider two
vectors $x,x'\in[0,1]^{E}$, such that for every $e\in E$, $\size{x_{e}-x'_{e}}\le\delta$.
Then, $F(x')-F(x)\ge\sum_{e\in E}(x'_{e}-x_{e})\cdot\partial_{e}F(x)-O\br{n^{3}\delta^{2}}\cdot f\br{OPT}$.
\end{lem}

\begin{lem}[Lemma 3.5 in \cite{DBLP:conf/focs/FeldmanNS11}] \label{lem:feldman-2} Consider a vector
$x\in\brq{0,1}{}^{E}$. Assuming $x_{e}\le a$ for every $e\in E$,
for every set $S\subseteq E$ it holds $F\br{x\lor\chr
S}\ge(1-a)f\br S$. 
\end{lem}

\begin{lem}[Lemma 3.6 in \cite{DBLP:conf/focs/FeldmanNS11}]\label{lem:feldman-3}
For every time $0\le t\le T$
and element $\ensuremath{e\in E},\ensuremath{y_{e}\br t\le1-\br{1-\delta}{}^{t/\delta}\le1-e^{-t}+O\br{\delta}}$.
\end{lem}

\begin{proof}[Proof of Lemma~\ref{lem:feldman-generalization}]
Applying Lemma \ref{lem:feldman-1} to vectors $y(t+\delta)$
and $y(t)$, we have 

\begin{align}
 & F\br{y(t+\delta)}-F(y(t)) \ge \nonumber  \\
\ge & \sum_{e\in E}\delta\cdot I_{e}\br t(1-y_e(t))\cdot\partial_{e}F(y(t))-O\br{n^{3}\delta^{2}})\cdot f(OPT) = \nonumber \\
= & \sum_{e\in E}\delta\cdot I_{e}\br t(1-y_e(t))\cdot\frac{F\br{y(t)\lor\chr j}-F(y(t))}{1-y_e(t)}-O(n^{3}\delta^{2})\cdot f(OPT) = \nonumber \\
= & \sum_{e\in E}\delta\cdot I_{e}\br t\cdot\brq{F(y(t)\lor\chr e)-F(y(t))}-O(n^{3}\delta^{2})\cdot f(OPT) \ge \nonumber \\
\ge & \sum_{e\in E}\delta\cdot x_{e}^{+}\brq{F(y(t)\lor\chr e)-F(y(t))}-O(n^{3}\delta^{2})\cdot f(OPT),\label{eq:cg-1-1}
\end{align}
where the last inequality follows from our choice of $I(t)$.
Further, we have $f^{+}(x^{+})=\sum_{A\subseteq E}\alpha_{A}f(A)$
for some set of values $\alpha_{A}$ satisfying $\sum_{A\subseteq E}\alpha_{A}=1$
and $\sum_{A\subseteq E:e\in A}\alpha_{A} \le x_{e}^{+}$. Thus, 

\begin{multline*}
\sum_{e\in E}x_{e}^{+}\brq{F(y(t)\lor\chr e)-F(y(t))} \ge \sum_{A\subseteq E}\alpha_{A}\sum_{e\in A}\brq{F(y(t)\lor\chr e)-F(y(t))} \ge \\
\ge\sum_{A\subseteq E}\alpha_{A}\brq{F(y(t)\lor\chr A)-F(y(t))}\ge\sum_{A\subseteq E}\alpha_{A}\brq{(e^{-t}-O(\delta))\cdot f(A)-F(y(t))} = \\
=(e^{-t}-O(\delta))\cdot f^{+}(x^{+})-F(y(t)),
\end{multline*}
where the second inequality follows from the fact that $F$ is concave
in all positive directions, and the third from Lemmas~\ref{lem:feldman-2}
and~\ref{lem:feldman-3}. Combining this with
(\ref{eq:cg-1-1}), and noting that $f^{+}\br{x^{+}}\ge f^{+}\br{OPT}=f\br{OPT}$,
we finally obtain $F(y(t+\delta))-F(y(t))\ge\delta\cdot\brq{e^{-t}\cdot f^{+}(x^{+})-F(y(t))}-O(n^{3}\delta^{2})\cdot f^{+}(x^{+}).$
\end{proof}
\section{Single client routine for BMUMD}
\label{sec:single-client}

Consider the moment when we have decided to serve agent $i$. Note that some items 
from ${\cal J}_{i}$ might have already been blocked.
We are supplied with the vector $(x_{c,p})_{c \in {\cal J}_i,\, p \in\cal B}$
and we assume that the variables for blocked items are set to 0.

\subsection{First attempt}

Imagine that with probability $\frac{1}{2}$ we discard item $c$ and it 
does not go to the menu.
Then we set its price to $p$ with probability $x_{c,p}$ and add it to the menu of client $i$.
Note that $\sum_{p}x_{c,p}\leq1$, so it is possible that we assign no price to the item, and in this case we discard it.
Let us emphasize that this happens independently to the initial coin toss.
\begin{lem}
\label{lem:one-client-first}
With probability at least 
\[
\frac{1}{4}\sum_{p}x_{c,p}\cdot\pr{v_{c}\geq p}
\]
$c$ will be the item chosen from the menu by the client $i$.
\end{lem}

\begin{proof}
Consider another item $d$. With probability $\frac{1}{2}$ it goes
into the menu and with probability $x_{d,p}$ we set its price to $p$.
This price is acceptable by the agent (gives non-negative utility)
with probability $\pr{v_{d}\geq p}$.
From union-bound we can say that the probability of any such
event over all items $d\neq c$
it is at most
\[
\sum_{d\neq c}\frac{1}{2}\sum_{p}x_{c,p}\cdot\pr{v_{c}\geq p}\leq\frac{1}{2},
\]
therefore with probability at least $\frac{1}{2}$ no item $d\neq c$
is offered with a non-negative utility price.

With probability $\frac{1}{2}\sum_{p}x_{c,p}\cdot\pr{v_{c}\geq p}$
item $c$ ends up in the menu with a non-negative utility price.
Since this is independent from the event above, we see that with probability
at least $\frac{1}{4}\sum_{p}x_{c,p}\cdot\pr{v_{c}\geq p}$ item
$c$ is the only reasonable choice for the client, so we are sure $c$ will be chosen.
\end{proof}

\subsection{Almost perfect menu}
\global\long\def\evx#1#2#3{\mathbf{X}^{#1}_{#2,#3}}

In the previous section we have guaranteed that the probability of an item $c$ with price $p$ being at the top of the menu is proportional to $x_{c,p}\cdot\pr{v_{c}\geq p}$ with the ratio within $[\frac{1}{4}, \frac{1}{2}]$.
Now we are going to compress this interval to $[\frac{1}{4}, \frac{1}{4} + \eps]$.
Note that we cannot simply scale down the variables $(x_{c,p})$ because decreasing the value of $x_{c,p}$ may increase chances of winning for another item.

Recall that the vector $\mathbf{x} = (x_{c,p})$ describing randomized menu for client $i$ must satisfy following constraints and we call it a \emph{menu-vector}.
\begin{eqnarray*}
 & \sum_{p}x_{c,p}\leq1 & \forall c\in {\cal J}_{i}\\
 & \sum_{c\in{\cal J}_{i}}\sum_{p}x_{c,p}\cdot\pr{v_{c}\geq p}\leq1
\end{eqnarray*}

Given menu-vector $\mathbf{x}$, we construct the menu as follows.
Independently for each item $c$ we choose price $p$ with probability $x_{c,p}$ and discard the item with probability $1 - \sum_p x^t_{c,p}$.
Then the client reveals their utilities for each item.
We define $\evx {} c p$ to be the event of the item $c$ with price $p$ being at the top of the menu.

\singleclient*
\begin{proof}
We begin with $\mathbf{x}^0 = \frac{1}{2}\mathbf{x}$ as the first approximation.
We are going to construct a series o menu-vectors $\mathbf{x}^t$, each time decreasing the discrepancy, that converges to $\mathbf{y}$ in $O(1/\eps^2)$ steps.

Let us define $q(c,p) =\frac{1}{4} x_{c,p}\cdot\pr{v_{c}\geq p}$.
In order to construct $\mathbf{x^{t+1}}$ we compute set
$D^t = \{(c,p)\, :\, \pr{\evx {t} c p} > (1+2\eps)\cdot q(c,p)\}$
and scale down the variables according to the formula

\begin{equation*}
    x^{t+1}_{c,p}= 
\begin{cases}
    x^t_{c,p} & \text{if } (c,p) \not\in D^t,\\
    (1-\eps)\cdot x^t_{c,p} & \text{if } (c,p) \in D^t.
\end{cases}
\end{equation*}

We will take advantage of the \emph{coupling} technique to analyze  deltas between $\pr{\evx {t} c p}$and $\pr{\evx {t+1} c p}$.
The idea is to construct a common probabilistic space where events $(\evx t c p)$ and $(\evx {t+1} c p)$ are correlated.
We decide on each item $c$ independently by setting its price to $p$ with probability $x^t_{c,p}$ and discarding $c$ with probability $1 - \sum_p x^t_{c,p}$.
When the price gets fixed we check whether $(c,p) \in D^t$ and, if yes, we discard $c$ with probability $\eps$ independently to the previous choices.
This procedure is equivalent to choosing prices with respect to $(x^{t+1}_{c,p})$.

Let $\mathcal{E}^t(c_1,p_1,c_2,p_2)$ be an event indicating that in step $t$ the pair $(c_1,p_1)$ was at the top of the menu, then $c_1$ got discarded in the second phase, and $(c_2,p_2)$ is at the top of the menu in step $t+1$.
One can think of this as of transferring the victory from $(c_1,p_1)$ to $(c_2,p_2)$.
Note that events $\mathcal{E}^t(c_1,p_1,c_2,p_2)$ make sense only in the coupled probabilistic space but nevertheless we can use them to estimate the probabilities.
Namely, we have

\begin{equation}
\label{eq:dump-choice}
    \pr{\evx {t+1} c p}= 
\begin{cases}
    \quad\quad\quad\pr{\evx t c p} + \sum_{\myatop{p_1}{c_1\neq c}}\pr{\mathcal{E}^t(c_1,p_1,c,p)} & \text{if } (c,p) \not\in D^t,\\
    (1-\eps)\cdot \pr{\evx t c p} + \sum_{\myatop{p_1}{c_1\neq c}}\pr{\mathcal{E}^t(c_1,p_1,c,p)} & \text{if } (c,p) \in D^t.
\end{cases}
\end{equation}

If $\mathcal{E}^t(c_1,p_1,c_2,p_2)$ has occurred, then both $(c_1,p_1),\,(c_2,p_2)$ must have been included in the menu with non-negative utilities and $(c_1,p_1)$ must have been discarded in the second phase with probability $\eps$.
Since we can only decrease $x^t_{c,p}$, we have $x^t_{c,p} \le x^0_{c,p} = \frac{1}{2}x_{c,p}$ and

\begin{eqnarray}
\label{eq:dump-estimate}
& \pr{\mathcal{E}^t(c_1,p_1,c_2,p_2)} &\le \eps \cdot x^t_{c_1,p_1}\cdot \pr{v_{c_1}\geq p_1}\cdot x^t_{c_2,p_2}\cdot \pr{v_{c_2}\geq p_2} \le \nonumber \\
&&\le \frac{\eps}{4} \cdot x_{c_1,p_1}\cdot \pr{v_{c_1}\geq p_1}\cdot x_{c_2,p_2}\cdot \pr{v_{c_2}\geq p_2}, \nonumber\\
& \sum_{\myatop{p_1}{c_1\neq c}}\pr{\mathcal{E}^t(c_1,p_1,c,p)} &\le \frac{\eps}{4} \cdot x_{c,p}\cdot \pr{v_{c}\geq p}\cdot\sum_{\myatop{p_1}{c_1\neq c}} x_{c_1,p_1}\cdot \pr{v_{c_1}\geq p_1} \le \nonumber\\
&&\le \frac{\eps}{4} \cdot x_{c,p}\cdot \pr{v_{c}\geq p} = \eps\cdot q(c,p).
\end{eqnarray}

We are ready to formulate sufficiently tight bounds on deltas. 

\begin{enumerate}[label=(\alph*)]
\item $ \pr{\evx t c p}  \le \pr{\evx {t+1} c p} \le \pr{\evx t c p} + \eps\cdot q(c,p)$ for $(c,p)\not\in D^t$,
\item $\pr{\evx t c p} -  2\eps\cdot q(c,p) \le \pr{\evx {t+1} c p} \le \pr{\evx t c p} - 2\eps^2\cdot q(c,p)$ for $(c,p)\in D^t$,
\item the invariant $q(c,p)\le \pr{\evx t c p}$ is being maintained.
\end{enumerate}

Combining the formulas (\ref{eq:dump-choice}) and (\ref{eq:dump-estimate}) entails the property (a) directly. 
The left side of property (b) follows from
\begin{eqnarray*}
&\pr{\evx {t+1} c p} &\ge \pr{\evx t c p} - \eps\cdot \pr{\evx t c p} \ge \pr{\evx t c p} -\eps\cdot x^t_{c,p}\cdot \pr{v_{c}\geq p} \ge \\
&& \ge \pr{\evx t c p} -\eps\cdot x^0_{c,p}\cdot \pr{v_{c}\geq p} = \pr{\evx t c p} - \frac{\eps}{2}\cdot x_{c,p}\cdot \pr{v_{c}\geq p} = \\
&&= \pr{\evx t c p} -2\eps\cdot q(c,p).
\end{eqnarray*}
To handle the right side, recall that $ \pr{\evx t c p} > (1+2\eps)\cdot q(c,p)$ for $(c,p) \in D^t$, so
\[ \pr{\evx {t+1} c p} \le \pr{\evx t c p} - \eps\cdot \pr{\evx t c p} + \eps\cdot q(c,t) \le \pr{\evx t c p} - (\eps+2\eps^2 -\eps)\cdot q(c,p). \]

Finally we prove the invariant.
By Lemma~\ref{lem:one-client-first} we know that the property (c) is satisfied for $\mathbf{x}^0$, so we can proceed by induction.
The value of $\pr{\evx {t+1} c p}$ is smaller than $\pr{\evx t c p}$
only for $(c,p) \in D^t$.
Combining definition of $D^t$ with property (b) ensures that $\pr{\evx {t+1} c p} \ge q(c,p)$.

To finish the whole argument, observe that the probability can only increase when
$ \pr{\evx t c p} \le (1+2\eps)\cdot q(c,p)$ and the delta is bounded by $\eps\cdot q(c,p)$.
Therefore, as soon as the value of $\pr{\evx t c p}$
drops below $(1+3\eps)\cdot q(c,p)$ it remains there until the end of the procedure.
The values above this threshold are being truncated by $2\eps^2\cdot q(c,p)$ in each step, so after $O(1/\eps^2)$ steps all values $\pr{\evx t c p}$ lie respectively in $[q(c,p), (1+3\eps)\cdot q(c,p)]$, which is contained within $\big[\frac{1}{4}x_{c,p}\cdot\pr{v_{c}\geq p}, (\frac{1}{4} + \eps)\cdot x_{c,p}\cdot\pr{v_{c}\geq p}\big].$
\end{proof}
\section{Combining matroid and knapsack constraints}
\label{sec:knapsack}
We consider optimization over the knapsack
constraint, where each element $e$ is assigned size $s_e \in\brq{0,1}$ and a~set $X \subseteq E$ is considered independent as long as $\sum_{e\in X}s_{e}x_{e}\leq1$.
The
\emph{knapsack polytope} is given by ${\cal P}\br{{\cal I}}=\setst{x\in\mathbb{R}_{\geq0}^{E}}{\sum_{e\in E}s_{e}x_{e}\leq1}$.

We shall call such a constraint system a \textbf{bounded knapsack}, if we
additionally have $s_{e}\in\brq{0,\frac{1}{2}}$.
As for matroids, we illustrate the controller mechanism with a contention resolution scheme.

\subsection{A controller mechanism for the bounded knapsack constraint}
\label{sec:knapsack-bounded}
\begin{thm}
\label{thm:knapsack-cr}
There exists a random-order CR scheme for a bounded knapsack
with $c=\frac{1}{3}$.
\end{thm}

\begin{algorithm}
\caption{\label{alg:randomselection-1}Random order contention resolution scheme
for a knapsack}

\begin{algorithmic}[1]

\STATE $S\leftarrow\emptyset$

\STATE for each element $e\in E$ choose randomly a point from interval
$I=[0,1]$; call it the controller of $e$, and denote it by
$C_{e}$

\STATE \textbf{for} each $e\in E$ in random order \textbf{do}

\STATE ~~~\textbf{if $e\notin R\br x$ then }

\STATE ~~~~~~continue

\STATE ~~~\textbf{if} $(e, C_{e})$ has not been blocked\textbf{ then}

\STATE ~~~~~~$S\leftarrow S \cup \{e\}$

\STATE ~~~\textbf{~~~}randomly choose $2\cdot s_{e}$
mass from available points of interval $I$, and block it

\RETURN $S$

\end{algorithmic} 
\end{algorithm}

\paragraph{Implementation}
The controller $C_e$ is given by a random point from $I = [0,1]$.
When element $e$ is taken into the solution, it blocks $2\cdot s_{e}$ random mass from the non-blocked subset of $I$, or blocks everything if the remaining mass is less then $2\cdot s_{e}$.

Some explanation is necessary for this blocking procedure as we cannot implement sampling a random subset over real numbers.
However, the only property that we require is that when we have mass $\ell$ of available points and we sample mass $s \le \ell$, then the probability of hitting any particular point is $\frac{s}{\ell}$.

We can implement such a sampling by fixing a mapping from the set of available points to a circle with circumference $\ell$, choosing a point $x$ at the circle uniformly at random, and blocking the interval of length~$s$ starting from $x$ clockwise.
If we use a natural mapping that glues intervals of available points, then the number of these intervals will stay proportional to $|E|$.

\paragraph{Correctness}
Let $I^t \subseteq I$ denote the set of available points
at the beginning of step $t$.
We argue that $S$ is always an independent set, i.e.,
$\sum_{e\in S} s_e \le 1$.
Accepting element $f$ leads to removal of $2\cdot s_f$ mass from $I^t$.
This means that as long as there are available points in $I$, i.e., $\size{I^t} > 0$, the solution satisfies $\sum_{e\in S} s_e \le \frac{1}{2}$.
And since $s_e \le \frac{1}{2}$, we can add $e$ to the solution without breaking the constraint.

\begin{lem}
\label{lem:knapsack-blocking-pr}
The characteristic sequences for the bounded knapsack constraint
are 2-bounded.
\end{lem}
\begin{proof}
We choose an element $f\neq e$ with probability $\frac{1}{n-t}$ and it turns
out to exist in $R\br x$ with probability $x_f$.
In contrary to the matroid argument, we additionally take into account the probability that the controller assigned to $f$ does not get blocked.
This happens with probability at most $\size{I^{t}}$ because the controller $C_{f}$ has to belong to the leftover available set $I^{t}$.
Further, element $f$ causes removal of $2\cdot s_{f}$
mass from $I^{t}$. 
The probability that the point $C_{e}$ gets blocked is 
\[
\min\br{\frac{2\cdot s_{f}}{\size{I^{t}}},1}.
\]
Combining these arguments, we can estimate the probability of $(e,\, C_e)$ getting blocked as follows:
\begin{align*}
 & \frac{1}{n-t}\sum_{f\neq e}x_f\cdot\size{I^{t}}\cdot\min\br{\frac{2\cdot s_{f}}{\size{I^{t}}},1} \le \\
\leq & \frac{1}{n-t}\sum_{f\neq e}x_{f}\cdot\size{I^{t}}\cdot\frac{2\cdot s_{f}}{\size{I^{t}}} = \\
= & \frac{1}{n-t}\sum_{f\neq e}x_{f}\cdot2\cdot s_{f} \le \\
\leq & \frac{2}{n-t},
\end{align*}
where the last inequality follows from the definition of the knapsack polytope.
Hence, the probability of a blocking event for $(e, C_e)$ at step $t$ is at most $\frac{2\cdot Y_e^t}{n-t}$.
\end{proof}

Theorem \ref{thm:knapsack-cr} follows immediately from Lemmas~\ref{lem:cr} and~\ref{lem:knapsack-blocking-pr}.

\subsection{Reduction to the bounded case}

We consider now the general variant of the knapsack constraint
with $s_{e}\in\brq{0,1}$.
We divide the elements into $E_{big} = \{s_e > \frac{1}{2} \,|\, e\in E\}$
and $E_{small} = \{s_e \le \frac{1}{2} \,|\, e\in E\}$.
With probability $\frac{1}{2}$ we consider only big items and discard all small items,
and vice versa.
Whereas the controller mechanism for small items (i.e. the bounded knapsack) has been presented in Section~\ref{sec:knapsack-bounded}, the case with big items reduces to the uniform matroid $U_1$.

\begin{lem}
\label{lem:knapsack-reduction}
Let $\cal I$ be a knapsack constraint with only big items.
Then $U_1 \subseteq \cal I$ and $\frac{1}{2} \P{\cal I} \subseteq  \P{U_1}$.
\end{lem}
\begin{proof}
The first claim is obvious as every singleton set is independent in $\cal I$.
The second one says that any vector from
${\cal P}\br{{\cal I}}=\setst{x\in\mathbb{R}_{\geq0}^{E}}{\sum_{e\in E}s_{e}x_{e}\leq1}$
sums to at most 1 after scaling by $\frac{1}{2}$, what is also straightforward as
\[
\sum_{e\in E} \frac{x_{e}}{2} \leq \sum_{e\in E}s_{e}x_{e} \leq1.
\]
\end{proof}

When the constraint system contains $q$ knapsacks,
we toss a coin independently $q$ times,
deciding for each knapsack whether we discard  its small or big items.
The probability that a given element is not discarded in the end clearly equals $1/2^q$.
However, we need a more careful argument to analyze expected revenue in submodular optimization.

\begin{lem}
\label{lem:knapsack-submodular}
Consider a non-negative submodular function $f$ over the ground set $E$ with $q$ partitions $E = E_1^i \uplus E_2^i$ for $i \in [q]$.
Let $\rho_1, \rho_2, \dots \rho_q$ be independent random variables equal to 0 or 1 with probability~$\frac 1 2$.
Then for every $A \subseteq E$ it holds
\[
\ex{f(A \cap E_{\rho_1}^1 \cap E_{\rho_2}^2 \cap \dots \cap E_{\rho_q}^q)} \ge \frac 1 {2^q} \cdot f(A).
\]
\end{lem}
\begin{proof}
From submodularity and non-negativity we have
\[
f(A \cap E_1^1) + f(A \cap E_2^1) \ge f\br{(A \cap E_1^1) \cup (A \cap E_2^1)} + f\br{(A \cap E_1^1) \cap (A \cap E_2^1)} = f(A) + f(\emptyset) \ge f(A),
\]
therefore $\ex{f(A \cap E_{\rho_1}^1)} \ge \frac{1}{2}\cdot f(A)$.
We iterate this argument, each time decreasing the bound by 2.
\end{proof}

\subsection{Results for knapsack and matroid constraints}

In this section we revisit the main results and briefly explain how to extend them to work with knapsack constraints.
For the sake of simplicity we do not optimize the probability of switching between small and big items and always set it to $\frac 1 2$.
Whereas some minor improvements in the approximation ratios are possible, the main message is that we can maintain the linear dependency on the number of matroids $k$ if the number of knapsack constraints is $O(1)$,
matching the results from~\cite{DBLP:conf/soda/FeldmanSZ16}. \marr{sprawdzic czy chekuriego tez nie matchujemy}

In the theorem below note that the decision whether an element gets discarded is made in advance depending on its sizes in knapsack constraints.
We can assume to know them in advance, before discovering the existence of the element, as this information is a part of the constraint structure.

\begin{thm}
\label{thm:combined-cr}
There exists a random-order CR scheme for intersection of $k$ matroids and $q$ knapsacks
with $c=\frac{1}{2^{q+1}} \cdot \frac 1 {k+2q+1}$.
\end{thm}
\begin{proof}
We are given a vector $x$ from the polytope of the constraint system.
For each knapsack constraint we independently choose whether we consider only big or only small elements with probability $\frac 1 2$.
The knapsack constraints in which we consider only small items become bounded knapsacks, described in Section~\ref{sec:knapsack-bounded}.
The other knapsack constraints are replaced with matroid $U_1$.
For the discarded elements we set $x'_e = 0$
and for the rest $x'_e = \frac{x_e}{2}$.
Lemma~\ref{lem:knapsack-reduction} guarantees that $x'$ belongs to the polytope of the new constraint system.

We have managed to reduce the general case to an intersection of $k + q'$ matroids and $q - q'$ bounded knapsacks for $0 \le q' \le q$.
The controller mechanism for those have been described in Sections~\ref{sec:matroid-cr},~\ref{sec:knapsack-bounded}
and we combine them with Lemma~\ref{lem:deltas-multi} obtaining $\lam$-bounded characteristic sequences with $\lam = k + q' + 2(q- q') \le k + 2q$.
Lemma~\ref{lem:cr} says that the constructed CR-scheme accepts each element with probability at least $\frac {x'_e}{k+2q+1}$.
Since $\ex{x'_e} = \frac{x_e}{2^{q+1}}$, this finishes the construction.
\end{proof}

\begin{thm}
Bayesian multi-parameter unit-demand mechanism design (BMUMD) over an intersection of $k$ matroids and $q$ knapsacks admits a
$2^{q+1} \cdot (k+2q+4+\eps)$ approximation algorithm.
\end{thm}
\begin{proof}
We inject reduction from
Theorem~\ref{thm:combined-cr} to the proof of Theorem~\ref{thm:bmumd-main}.
\end{proof}

\begin{thm}
\label{thm:combined-submodular}
Maximization of a non-negative submodular function $f$ over an intersection of $k$ matroids and $q$ knapsacks admits a
$2^{q+1} \cdot (k+2q+1) \cdot (\sqrt{e}+\eps)$ approximation algorithm.
\end{thm}
\begin{proof}
Let $\cal P$ be the polytope induced by the constraint system.
We execute the measured continuous greedy algorithm\marr{sprawdz continues czy sie nie pojawia jeszcze gdzie indziej} (\cite{DBLP:conf/focs/FeldmanNS11}, also see Lemma~\ref{lem:measured-greedy})
with $b=\frac 1 2$.
The returned vector $x$ satisfies 1) $x \in \frac{1}{2} \cal P$, and 2) $F(x) \ge \frac{1}{2\sqrt{e} + \eps} \cdot f(OPT)$.
Then we the apply the reduction from
Theorem~\ref{thm:combined-cr}
and discard a random subset of elements
(note that this time we do not need to scale $x$ by $\frac 1 2$ because we have guaranteed that $x \in \frac{1}{2} \cal P$).
For the sake of analysis we do not reveal the set of discarded items $D$, and we simulate the routine from Lemma~\ref{lem:submodular-abstract}
without any constraints on the discarded items.
In this setting, the returned solution $X$ satisfies $\ex{f(X)} \ge \frac{1}{k+2q+1}\cdot F(x)$.
Since the true solution is given by $X \setminus D$ and for all $A \subseteq E$ with random choice of $D$ we have $\ex{f(A \setminus D)} \ge \frac{1}{2^q}\cdot f(A)$ (Lemma~\ref{lem:knapsack-submodular}), we conclude that $\ex{f(X\setminus D)} \ge \frac{1}{2^q} \cdot \frac{1}{k+2q+1} \cdot F(x)$.
\end{proof}

\begin{thm}
Non-negative submodular stochastic probing over an intersection of $k$ matroids and $q$ knapsacks admits a
$2^{q+1} \cdot (k+2q+1) \cdot (\sqrt{e}+\eps)$ approximation algorithm.
\end{thm}
\begin{proof}
Let $\cal P$ be the polytope induced by the constraint system.
As in Section~\ref{sec:probing} we rely
on the upper bound from Lemma~\ref{lem:math-programming-bound}, i.e.,
$B = \max_{x\in\P{\Iin,\Iout}} f^{+}\br{x}$.
We  execute the measured continuous greedy algorithm with a stronger bound (Lemma~\ref{lem:measured-greedy})
and $b=\frac 1 2$.
The returned vector $x$ satisfies 1) $x \in \frac{1}{2} \cal P$, and 2) $F(x) \ge \frac{1}{2\sqrt{e} + \eps} \cdot B$.
We apply the reduction from
Theorem~\ref{thm:combined-cr} and then proceed
as in the proof of Theorem~\ref{thm:probing-main}.
The lower bound for the expected value of the objective function is the same as in Theorem~\ref{thm:combined-submodular}.
\end{proof}

\bibliographystyle{plain}
\bibliography{rocrscheme}

\end{document}